\documentclass[11pt,a4paper,reqno]{amsart}

\usepackage{hyperref}
\hypersetup{
  colorlinks   = true, 
  urlcolor     = blue, 
  linkcolor    = blue, 
  citecolor   = red 
}
\usepackage{changepage}
\usepackage{amsmath}
\usepackage{amsthm}
\usepackage{amsfonts}
\usepackage{amssymb}
\usepackage{array}

\usepackage{soul}

\usepackage{amsmath}
\usepackage{amsthm}
\usepackage{amsfonts}
\usepackage{amssymb}
\usepackage{array}
\usepackage[utf8]{inputenc}
\usepackage[margin=2.7cm]{geometry} 

\usepackage{graphicx} 
\usepackage{algorithm}
\usepackage{algorithmic}
\usepackage[T1]{fontenc}
\usepackage[utf8]{inputenc}
\usepackage{xcolor}
\usepackage{amsaddr}

\newtheorem{theorem}{Theorem}
\newenvironment{thmbis}
  {\addtocounter{theorem}{-1}%
   \begin{theorem}}
  {\end{theorem}}

\newtheorem{proposition}{Proposition}
\newtheorem{corollary}{Corollary}
\newtheorem{lemma}{Lemma}

\newtheorem{problem}{Problem}

\theoremstyle{definition}
\newtheorem{definition}{Definition}

\newtheorem{example}{Example}
\newtheorem{remark}{Remark}

\newcommand{\C}{\mathcal{C}}

\newcommand{\GL}{\mathrm{GL}}
\newcommand{\F}{\mathbb{F}}
\newcommand{\E}{\mathcal{E}}

\newcommand{\rk}{\mathrm{rk}}
\newcommand{\bs}{}

\newcommand{\Gab}{\mathrm{Gab}}
\newcommand{\Gal}{\mathrm{Gal}}

\newcommand{\Gr}{\mathrm{Gr}}

\newcommand{\rsu}{\mathrm{supp}}
\newcommand{\rs}{\mathrm{rs}}

\newcommand{\Fm}{\mathbb{F}_{q^m}}
\newcommand{\Tr}{\mathrm{Tr}_{\F_{q^m}/\F_q}}
\newcommand{\Hom}{\mathrm{Hom}}
\newcommand{\Fq}{\mathbb{F}_{q}}
\newcommand{\G}{\mathcal{G}}

\newcommand{\npmatrix}[1]{\left( \begin{matrix} #1 \end{matrix} \right)}

\title[Systematic Gabiduin codes and $q$-Cauchy matrices]{Systematic encoders for generalized Gabidulin codes and the $q$-analogue of Cauchy matrices}

\author[Alessandro Neri]{Alessandro Neri$^*$}
\address{Technical University of Munich, Germany}
\curraddr{}
\email{alessandro.neri@tum.de}
\thanks{$^*$Alessandro Neri was supported by the Swiss National Science Foundation through grants no. 169510 and 187711.}



\begin{document}
\maketitle

\begin{abstract}
 We characterize the generator matrix in standard form of generalized Gabidulin codes. 
The parametrization we get for the non-systematic part of this matrix coincides with the $q$-analogue of generalized Cauchy matrices, leading to the definition of \emph{$q$-Cauchy matrices}.
These matrices can be represented very conveniently and their representation allows to define new interesting subfamilies of generalized Gabidulin codes whose generator matrix is a structured matrix. In particular, as an application, we construct Gabidulin codes whose generator matrix is the concatenation of an identity block and a Toeplitz/Hankel matrix. In addition, our results allow  to give a new efficient criterion  to verify whether a rank metric code of dimension $k$ and length $n$ is a generalized Gabidulin code. This criterion is only based on the computation of the rank of one matrix and on the verification of the linear independence of two sets of elements and it requires  $\mathcal O(m\cdot F(k,n))$ field operations, where $F(k,n)$ is the cost of computing the reduced row echelon form of a $k \times n$ matrix. Moreover, we also provide a characterization of the  generator matrix in standard form of general MRD codes.
\end{abstract}

\section{Introduction}

Codes in the rank metric were introduced, independently, by Delsarte \cite{de78}, Gabidulin \cite{ga85a} and Roth \cite{ro91}, although a similar notion can be traced back to Bergman \cite{be69}. However, only in the last ten years have they significantly gained  interest, due to their application in network coding \cite{si08j,ko08}. Moreover, rank metric codes  have a plethora of different applications in communications and security. In addition to network coding, the applications proposed in the last 20 years concern cryptography \cite{lo10, ov08}, space-time coding and wireless communications \cite{ta98, lu03}, distributed storage \cite{silberstein2015error, ca17, ne18pmds}, authentication schemes \cite{ne17} and low-rank matrix completion \cite{mu17}. 

 As with codes in the Hamming metric, they are usually defined over a finite field $\F_q$, and in the linear case their important parameters are given by the length $n$, the dimension $k$ and the minimum distance $d$. Those parameters are related by an inequality that is as elegant as effective. This is the well-known Singleton bound, that holds in both the Hamming and rank metric. Hamming codes meeting this bound are called \emph{maximum distance separable (MDS)} codes. Their natural analogue in the rank metric is represented by \emph{maximum rank distance (MRD)} codes, that are defined analogously as codes that attain the Singleton bound with equality. Although it was proven that there are plenty of MRD codes that are linear over the extension field \cite{ne18, by18}, only few new families have been discovered recently \cite{sh16, lu15, ot17,  sheekey2019new, puchinger2017further} and some sporadic construction \cite{ho16, csajbok2018new, csajbok2018maximum, csajbok2018newbis, marino2019mrd, bartoli2019new}.

However, the  most studied and important construction of MRD codes is still the one proposed in the seminal works \cite{de78, ga85a, ro91}, and then generalized in \cite{ks05}. These codes are known as \emph{generalized Gabidulin codes}, and they represent the rank analogue of the well-known \emph{generalized Reed-Solomon (GRS) codes}. As GRS codes, generalized Gabidulin codes are evaluation codes. However,  they are defined over an extension field $\Fm$ of $\Fq$ and the evaluation is done on a particular subset of \emph{linearized polynomials} in $n$ points  that are linearly independent over $\F_q$. The structure of evaluation codes allowed the development of many efficient decoding algorithms in the last years \cite{si09p, wa13}.

In this framework, another analogy emerges regarding the generator matrices of these two families of codes. The canonical generator matrix of GRS codes is obtained by the evaluation of the canonical basis $\{1,x,\ldots, x^{k-1}\}$, that gives as a result the  \emph{weighted Vandermonde matrix}, a matrix given by the product of a Vandermonde with a non-singular diagonal matrix. The rank analogue of the weighted Vandermonde matrix is given by the $s$-Moore matrix. Such a matrix is the canonical generator matrix of a generalized Gabidulin code, obtained via the evaluation of the canonical basis $\{x, x^{q^s}, x^{q^{2s}}, \ldots, x^{q^{(k-1)s}} \}$.

There is another important generator matrix of GRS codes that is well-known in the literature. In 1985 Roth and Seroussi gave a characterization of the generator matrix in standard form for these codes, showing that GRS codes are in correspondence with \emph{generalized Cauchy matrices} (\cite{ro85}). The same characterization was also given independently by D\"ur in \cite{dur1987automorphism}. Explicitly, the generator matrix in standard form of a GRS code is given by $(I_k \mid X)$, where $I_k$ is the $k\times k $ identity matrix, and $X$ is a generalized Cauchy matrix. On the other hand, every matrix $(I_k \mid X)$, with $X$ a generalized Cauchy matrix, generates a GRS code. 

In this work we  give a characterization of the generator matrix in standard form of a generalized Gabidulin code, that up to now was  unknown. As a consequence, this also allows us to define a rank analogue of generalized Cauchy matrices, whose definition coincides with the $q$-analogue of generalized Cauchy matrices. This result is obtained making a wide use of properties of finite fields, in particular the trace map, and of some recent results appeared recently \cite{ho16,ne18}.

In addition to the theoretical result that almost completes the picture on the analogies between GRS and generalized Gabidulin codes, this has also a useful impact from a practical point of view. Using the structure of the rank analogue of a generalized Cauchy matrix, we derive a subfamily of generalized Gabidulin codes whose generator matrix is made by an identity block and a Toeplitz/Hankel block. From an application point of view, this new family of codes seems to be suitable for fast algorithms for erasure correction and syndrome decoding as well as for encoding. It is well-known, indeed, that the matrix-vector multiplication with a Toeplitz/Hankel matrix can be performed in a fast way. 

Moreover, from  the theoretical characterization obtained, we also derive a new criterion to determine whether a given code  is a generalized Gabidulin code. This new criterion is faster to compute than any other previously known. Indeed, for a given rank metric code of dimension $k$ and length $n$ over a finite field $\Fm$, it only requires $\mathcal O(m\cdot F(k,n))$ field operations, where $F(k,n)$ denotes the cost of computing the reduced row echeleon form of a $k \times n$ matrix.

The paper is structured as follows. In Section 2 we recall some basic properties of finite fields and in particular of the field trace map. We also briefly explain the main results on GRS codes and on their generator matrices. In Section 3 we introduce rank metric codes and give a recap on the most important results on MRD and generalized Gabidulin codes. In addition, some new results are presented that are preparatory for the rest of the paper. Moreover, we give a characterization of the generator matrix of general MRD codes, in the spirit of the well-known results for MDS codes.
Section 4 represents the main contribution of this work. Here we characterize the generator matrix in standard form of a generalized Gabidulin code. From this result we derive a new criterion for determining if a given rank metric code is a generalized Gabidulin code. This section can be also seen as the analogue of Roth and Seroussi  \cite{ro85} and D\"ur \cite{dur1987automorphism} works for the rank metric and it completes the picture on the generator matrices of generalized Gabidulin codes. In Section 5 we use our previous results for constructing subfamilies of Gabidulin codes with structured generator matrix. These codes have generator matrix in standard form with a Hankel or a Toeplitz non-systematic part, potentially very useful for applications. Finally, in Section 6 we summarize the work underlining our main contributions.

\section{Preliminaries}
Throughout the whole work, given a map $h:\mathcal X \rightarrow \mathcal Y$ and a subset $ \mathcal T\subseteq \mathcal Y$, we  denote by $h^{-1}(\mathcal T)$ the preimage of the set $\mathcal T$, i.e.
$$h^{-1}(\mathcal T)=\left\{ x \in \mathcal X \mid h(x) \in \mathcal T\right\}.$$
 In the same way, for a set  $\mathcal S\subseteq \mathcal X$, $h(\mathcal S)$ denotes the set of images of the elements in $\mathcal S$ through $h$, i.e.
$$h(\mathcal S)=\left\{ h(x) \mid x \in \mathcal S  \right\}.$$

\subsection{Trace over finite fields an its duality}
The following definitions and results can be found in any textbook on
finite fields, e.g.\ \cite{li94}.  We denote the finite field of
cardinality $q$ by $\F_q$. It is well-known that it exists if and only if $q$ is a prime
power. Moreover, if it exists, $\F_q$ is unique up to isomorphism. An
extension field of extension degree $m$ is denoted by $\F_{q^m}$.
An important property of finite fields is the existence of a primitive element. This means that there always exists $\alpha \in \F_q$ that is a generator of $\F_q^*$, i.e. 
$$\F_q=\{0\}\cup\{\alpha^i \mid 0 \leq i \leq q-2\}.$$
We now recall some basic theory on  finite
fields and the trace function.
It is well-known that the extension field $\Fm$ over $\F_q$ is a Galois extension
and $\mathrm{Gal}(\Fm/\F_q)$ is cyclic. One of its generators is given by the $q$-Frobenius automorphism $\theta$, defined as
$$\begin{array}{rcl} \theta:\Fm & \longrightarrow & \Fm \\
\alpha & \longmapsto & \alpha^q.
\end{array}$$

\begin{definition}
  Let $\F_q$ be a finite field and $\F_{q^m}$ be an extension
  field. For $\alpha \in \F_{q^m}$, the \emph{trace} of $\alpha$ over
  $\F_q$ is defined by
  $$\mathrm{Tr}_{\F_{q^m}/\F_q}(\alpha)  := \sum_{i=0}^{m-1}\theta^i(\alpha)= \sum_{i=0}^{m-1}\alpha^{q^i}.$$
\end{definition}

For every integer $0<s<m$ with $\gcd(m,s)=1$, we denote by $\varphi_s$
the map given by
 $$ \begin{array}{rcl}
   \varphi_s:\F_{q^m} &\longrightarrow & \F_{q^m} \\
   \alpha & \longmapsto & \theta^s(\alpha)-\alpha.
 \end{array}
$$
We will refer to the function 
$$\begin{array}{rcl} \Tr:\Fm & \longrightarrow & \F_q
\end{array}$$
as the \emph{trace map} of $\Fm/\F_q$.

The following result relates the trace map with the functions $\varphi_s$.

 \begin{lemma}\label{lem:trace}
   The trace map satisfies the following properties:
   \begin{enumerate}
   \item $\mathrm{Tr}_{\F_{q^m}/\F_q}(\alpha) \in \F_q$ for all
     $\alpha \in \F_{q^m}$.
   \item $\mathrm{Tr}_{\F_{q^m}/\F_q}$ is an $\F_q$-linear surjective
     transformation from $\F_{q^m}$ to $\F_q$.
 \item $\varphi_s$ is an $\F_q$-linear transformation from $\F_{q^m}$ to
   itself.
 \item For every $s$ coprime to $m$, $\varphi_s(\alpha)=0$ if and only
   if $\alpha \in \F_q$.
 \item (Additive Hilbert’s Theorem 90 for finite fields) $\ker (\mathrm{Tr}_{\F_{q^m}/\F_q})=\mathrm{Im}(\varphi_s)$ for
   every $s$ coprime to $m$ and has cardinality $q^{m-1}$.
 \end{enumerate}
\end{lemma}

\begin{proof}
 A partial proof of this result can be found in \cite[Chapter 2, Section 3]{li94}. For a complete proof we refer to \cite[Lemma 2]{ne18}.
\end{proof}

 The trace map has many important properties. One of them is that it can be used to define an isomorphism between $\Fm$ and $\Hom_{\Fq}(\Fm,\Fq)$.

\begin{definition}
 The $\F_q$-bilinear map defined as 
 $$ \begin{array}{rcl}
  \mathrm{tr}: \Fm \times \Fm & \longrightarrow & \F_q \\
  (\alpha,\beta) & \longmapsto & \mathrm{Tr}_{\F_{q^m}/\F_q}(\alpha\beta),
 \end{array}
 $$
 is called the \emph{trace form} of the extension $\Fm/\F_q$.
 \end{definition}
 Observe that for every $\alpha \in \Fm$, we can associate an $\F_q$-linear map $T_\alpha \in \Hom_{\F_q}(\Fm, \F_q)$, defined as
 $$\begin{array}{rcl} 
 T_{\alpha}:\Fm & \longrightarrow & \F_q \\
 \beta & \longmapsto & \Tr(\alpha\beta).
 \end{array}$$
\begin{theorem}\label{thm:dualisom}
 The trace form is a symmetric non degenerate $\F_q$-bilinear form. Moreover it induces a duality isomorphism given by
 $$\begin{array}{rcl}
   \Psi:\Fm & \longrightarrow & \Hom_{\F_q}(\Fm, \F_q)\\
   \alpha & \longmapsto & T_{\alpha}.
   \end{array}$$
\end{theorem} 

\begin{proof}
For the proof one can see \cite[Theorem 2.24]{li94}.
\end{proof}

The following results directly follow from Theorem \ref{thm:dualisom}.

\begin{corollary}\label{cor:neq}
 For every $\alpha \in \Fm^*$ the map $T_\alpha$ is non identically zero, and hence $\dim_{\F_q}(\ker (T_\alpha))=m-1$.
\end{corollary}
 
 \begin{corollary}\label{cor:Tlin}
 For every $\alpha,\beta \in \Fm$ and $\lambda, \mu \in \F_q$, we have
$$T_{\lambda \alpha+\mu \beta}=\lambda T_\alpha+\mu T_\beta.$$
\end{corollary}

Since the trace form induces a duality isomorphism, we can naturally define the notion of dual basis.
\begin{definition}
Given an $\F_q$-basis $\alpha_1,\ldots, \alpha_m$ of  $\Fm$ and $\beta_1,\ldots, \beta_m \in \Fm$, we say that $\beta_1, \ldots, \beta_m$ is a \emph{dual basis} of  $\alpha_1,\ldots, \alpha_m$ with respect to the trace form, if for all $i,j \in \{1,\ldots,m\}$
$$ \mathrm{tr}(\alpha_i,\beta_j)=\delta_{i,j}=\begin{cases} 1 & \mbox{ if } i=j \\
 0 & \mbox{ if } i\neq j.
\end{cases}$$
\end{definition}

\begin{remark}
Given an $\F_q$-basis $\alpha_1,\ldots, \alpha_m$ of  $\Fm$, the existence and uniqueness of its dual basis follow  by Theorem \ref{thm:dualisom} and the fact that $\Fm$ is a finite dimensional $\F_q$-vector space.
\end{remark}

\begin{lemma}\label{lem:dimT} For every $\alpha_1,\ldots,\alpha_k, \beta \in \Fm$, 
 $$\ker (T_{\alpha_1}) \cap \ldots \cap \ker (T_{\alpha_k}) \subseteq \ker (T_{\beta})$$
 if and only if $\beta \in \langle \alpha_1,\ldots, \alpha_k\rangle$.
\end{lemma}
\begin{proof}
 Suppose $\beta \in \langle \alpha_1,\ldots, \alpha_k\rangle $. By Corollary \ref{cor:Tlin}, we have $T_\beta=\lambda_1T_{\alpha_1}+\ldots+\lambda_kT_{\alpha_k}$. Hence, if $x\in \ker (T_{\alpha_1}) \cap \ldots \cap \ker (T_{\alpha_k})$, then
$$T_\beta(x)=\lambda_1T_{\alpha_1}(x)+\ldots+\lambda_kT_{\alpha_k}(x)=0+\ldots+0=0,$$
and therefore $x \in \ker(T_{\beta})$.

On the other hand, suppose  $\beta \notin \langle \alpha_1,\ldots, \alpha_k\rangle $. Let $s:=\dim_{\F_q} \langle \alpha_1,\ldots, \alpha_k\rangle$. Without loss of generality we can assume that $ \langle \alpha_1,\ldots, \alpha_k\rangle =  \langle \alpha_1,\ldots, \alpha_s\rangle$. Now, complete $\alpha_1,\ldots,\alpha_s, \beta$ to an $\F_q$-basis $ \alpha_1,\ldots,\alpha_s,\beta, \gamma_1,\ldots\gamma_{m-s-1}$ of $\Fm$ and consider its dual basis with respect to the trace form $ \tilde{\alpha}_1,\ldots,\tilde{\alpha}_s,\tilde{\beta}, \tilde{\gamma}_1,\ldots \tilde{\gamma}_{m-s-1}$. Therefore, $T_{\alpha_i}(\tilde{\beta})=0$ for every $i=1,\ldots,s$ and $T_{\beta}(\tilde{\beta})=1$, i.e. $$\tilde{\beta} \in \ker (T_{\alpha_1}) \cap \ldots \cap \ker (T_{\alpha_k}) \setminus \ker (T_{\beta}).$$

\end{proof}

\begin{proposition}\label{thm:kerT} For every $\alpha_1,\ldots,\alpha_k \in \Fm$, 
 $$ \dim_{\F_q}(\ker (T_{\alpha_1}) \cap \ldots \cap \ker (T_{\alpha_k}))=m-\dim_{\F_q}\langle \alpha_1,\ldots, \alpha_k \rangle.$$
\end{proposition}
\begin{proof}
Let $s:=\dim_{\F_q} \langle \alpha_1,\ldots, \alpha_k\rangle$. Without loss of generality we can suppose  $ \langle \alpha_1,\ldots, \alpha_k\rangle =  \langle \alpha_1,\ldots, \alpha_s\rangle$. By Lemma \ref{lem:dimT} we have $\ker (T_{\alpha_{s+1}}), \ldots,\ker(T_{\alpha_k}) \supseteq \ker (T_{\alpha_1}) \cap \ldots \cap \ker (T_{\alpha_s}) $, and hence 
$$ \ker (T_{\alpha_1}) \cap \ldots \cap \ker (T_{\alpha_k})=\ker (T_{\alpha_1}) \cap \ldots \cap \ker (T_{\alpha_s}).$$
Therefore, it is enough to prove the statement when $\alpha_1,\ldots, \alpha_k$ are linearly independent over $\F_q$. We use induction on $k$. If $k=1$ then $\dim_{\F_q}(\ker(T_{\alpha_1}))=m-1$ by Corollary \ref{cor:neq}.

Suppose now that the statement is true for $k-1$, i.e. $$\dim_{\F_q}(S)=m-k+1,$$
where $S:=\ker (T_{\alpha_1}) \cap \ldots \cap \ker (T_{\alpha_{k-1}})$. Then, by  Lemma \ref{lem:dimT}, $S\not\subseteq\ker (T_{\alpha_k})$, i.e. $S +\ker (T_{\alpha_k})=\Fm$. Therefore,
\begin{align*}
\dim_{\F_q}(S\cap \ker (T_{\alpha_k})) &= \dim_{\F_q}(S)+\dim_{\F_q}(\ker (T_{\alpha_k})) - \dim_{\F_q}(S+\ker (T_{\alpha_k}))\\
                     &=m-k+1+m-1-m \\
                     &= m-k.
\end{align*}
\end{proof}

Now, let $\alpha_1,\ldots, \alpha_k \in \Fm$ be $\F_q$-linearly independent and complete them to a basis $\alpha_1,\ldots,\alpha_m$ of $\Fm$. Let $\beta_1,\ldots, \beta_m$ be its dual bases. Then for every $i=1,\ldots,k$ we have $T_{\alpha_i}(\beta_j)=0$ for every $j=k+1,\ldots,m$, i.e. $\beta_j \in \ker(T_{\alpha_1})\cap\ldots\cap \ker(T_{\alpha_k})$. Moreover, by Proposition \ref{thm:kerT}, we get $\dim_{\F_q}(\ker(T_{\alpha_1})\cap\ldots\cap \ker(T_{\alpha_k}))=m-k$, and hence
$$ \ker(T_{\alpha_1})\cap\ldots\cap \ker(T_{\alpha_k}) =\langle \beta_{k+1}, \ldots, \beta_m \rangle. $$
We can now define the trace-orthogonal space of a subspace as follows.
 
\begin{definition}
 Let $ S:=\langle \alpha_1,\ldots, \alpha_k \rangle$ be an $\F_q$-subspace of $\Fm$. Then the \emph{trace-orthogonal} space of $S$ is defined as the $\F_q$-subspace 
 $$ S^{\times}:= \ker(T_{\alpha_1})\cap\ldots\cap \ker(T_{\alpha_k}). $$
\end{definition}

\begin{proposition}
 The subspace $S^{\times}$ is well-defined, i.e. it does not depend on the choice of the set of generators.
\end{proposition}

\begin{proof}
 Let $\{\alpha_1,\ldots, \alpha_k\}$ and $\{\alpha_1',\ldots, \alpha_t'\}$ be two sets of generators for a subspace $S$. We want to prove that $\ker(T_{\alpha_1})\cap\ldots\cap \ker(T_{\alpha_k})=\ker(T_{\alpha_1'})\cap\ldots\cap \ker(T_{\alpha_t'})$. For every $i=1,\ldots,k$, $\alpha_i \in \langle \alpha_1',\ldots,\alpha_t' \rangle$ and therefore, by Lemma \ref{lem:dimT},  it holds that $\ker(T_{\alpha_i})\supseteq \ker(T_{\alpha_1'})\cap\ldots\cap \ker(T_{\alpha_t'})$. Hence, 
$$ \ker(T_{\alpha_1})\cap\ldots\cap \ker(T_{\alpha_k})\supseteq\ker(T_{\alpha_1'})\cap\ldots\cap \ker(T_{\alpha_t'}).$$
The opposite inclusion is analogous.
\end{proof}

We already know the relation between the image of the map $\varphi_s$ and the kernel of the trace map (see Lemma \ref{lem:trace}). The following Lemma characterizes the preimage of any element under the map $\varphi_s$.

\begin{lemma}\label{lem:preimage}
 Let $\alpha \in \Fm$ and $s$ a positive integer coprime to $m$. Then 
\begin{enumerate}
\item $$|\varphi_s^{-1}(\{\alpha\})| = \begin{cases}q & \mbox{ if } \alpha \in \ker(\Tr) \\ 
0 & \mbox{ if } \alpha \notin \ker(\Tr)
\end{cases}$$
\item Let $\alpha \in \ker(\Tr)$. If $x_1, x_2 \in \varphi_s^{-1}(\{\alpha\})$, then $x_1-x_2 \in \F_q$, or equivalently, there exists an $x \in \Fm$ such that
$$\varphi_s^{-1}(\{\alpha\})=\left\{ x+\lambda\mid \lambda \in \Fq \right\}.$$

Moreover such an $x$ is of the form

$$ x= -\frac{1}{\Tr(\gamma)} \sum_{i=0}^{m-2}\left(\sigma^{i+1}(\gamma)\sum_{j=0}^i(\sigma^j(\alpha))\right).$$
where $\gamma \in \Fm$ is such that $\Tr(\gamma)\neq 0$, and $\sigma:=\theta^s$.
\end{enumerate}
\end{lemma}

\begin{proof}
\begin{enumerate}
\item If $\alpha \notin \ker(\Tr)$, then, by part (5) of Lemma \ref{lem:trace}  we have $\varphi_s^{-1}(\{\alpha\})=\emptyset$. On the other hand, if $\alpha \in \ker(\Tr)$, then $|\varphi_s^{-1}(\{\alpha\})|=|\ker(\varphi_s)|$, since $\varphi_s$ is an $\F_q$-linear map. By part (2) of Lemma \ref{lem:trace}
$$q^{m-1}=|\ker(\Tr)|=|\mathrm{Im}(\varphi_s)|=\frac{|\Fm|}{|\ker(\varphi_s)|},$$
and therefore we get $|\varphi_s^{-1}(\{\alpha\})|=q$.
\item For the first part, let $x_1,x_2\in \varphi_s^{-1}(\{\alpha\})$. Hence,
$\varphi_s(x_1)-\varphi_s(x_2)=0,$ and by linearity of $\varphi_s$, we get $\varphi_s(x_1-x_2)=0$. By part (4) of Lemma \ref{lem:trace} we get $x_1-x_2 \in \F_q$. Finally, showing that $\varphi_s(x)=\alpha$ is a straightforward computation.
\end{enumerate}
\end{proof}

We conclude this section with a useful result on the linear independence of preimages of $\varphi_s $.

\begin{lemma}\label{lem:preindip}
Let $\alpha_1,\ldots, \alpha_k\in \ker(\Tr)$, $s$ be a positive integer coprime to $m$ and $\sigma:=\theta^s$. Suppose moreover that $\beta_1,\ldots,\beta_k \in \Fm$ are such that $\sigma(\beta_i)-\beta_i=\alpha_i$. Then, the elements $\alpha_1,\ldots, \alpha_k$ are linearly independent over $\F_q$ if and only if the elements $1,\beta_1,\ldots, \beta_k$ are linearly independent over $\F_q$. 
\end{lemma}

\begin{proof}
 Suppose $\lambda_1,\ldots,\lambda_k\in \F_q$ and consider the sum
 $$\sum_i \lambda_i\alpha_i=\sum_i \lambda_i(\sigma(\beta_i)-\beta_i)= \sigma(\sum_i \lambda_i\beta_i)-\sum_i \lambda_i\beta_i.$$
This means that a non-trivial combination of the $\alpha_i$'s is zero if and only if a non-trivial combination of the $\beta_i$'s belongs to $\ker \varphi_s$. This is equivalent, by part (4) of Lemma \ref{lem:trace}, to $\sum_i \lambda_i\beta_i \in \F_q$, i.e.  $1,\beta_1,\ldots, \beta_k$ are linearly dependent over $\F_q$. 
\end{proof}

\subsection{GRS codes and Generalized Cauchy Matrices}

In classical coding theory the most studied and well-known class of codes is definitely represented by generalized Reed-Solomon codes. These codes were introduced in \cite{rs60} and through the years were deeply studied by many authors. Their importance is due to the fact that they are maximum distance separable, and possess very fast algorithms for their encoding and decoding procedures \cite{gu98, ko03}.
In this section we are going to briefly describe them, focusing in particular on their generator matrices.

 Let $n$ be a positive integer. The Hamming distance $d_H$ on $\Fq^n$ is defined as
$$\begin{array}{rcl}
d_H: \F_q^n \times \F_q^n : & \longrightarrow & \mathbb N \\
(u,v) & \longmapsto & |\{i \mid u_i \neq v_i \}|.
\end{array}$$
It is well-known that $d_H$ defines indeed a metric on $\F_q^n$. With this metric, classical coding theory was developed in the last 70 years, focusing on many different classes of codes. In this section we will only consider linear codes.

\begin{definition}
 Let $0<k \leq n$ be two positive integers. A \emph{linear code} $\C$ of dimension $k$ and length $n$ over a finite field $\F_q$ is a $k$-dimensional $\F_q$-subspace of $\F_q^n$ equipped with the Hamming distance. The \emph{minimum distance} of $\C$ is the integer
$$d_H(\C):= \min\left\{ d_H(u,v) \mid u,v \in \C, u\neq v\right\}.$$
A matrix $G \in \F_q^{k\times n}$ is called a \emph{generator matrix} for the code $\C$ if $\C=\rs(G)$, where $\rs(G)$ denotes the subspace generated by the rows of $G$, called the \emph{row space} of $G$.
\end{definition}

It is well known that the minimum distance $d$ of any linear code of dimension $k$ and length $n$ satisfies the following inequality:
$$d\leq n-k+1.$$
This bound is known as Singleton bound \cite{si64} and codes meeting it with equality are called
 called  \emph{maximum distance separable (MDS) codes}.

Among all the possible generator matrices of an MDS code, there exists one in a special form. Indeed, it is easy to verify that every MDS code of length $n$ and dimension $k$ has a generator matrix of the form $G=(I_k \mid X)$, where $X \in \Fq^{k\times (n-k)}$ and $I_k$ denotes the $k \times k$ identity matrix. Such a generator matrix is said to be in  \emph{standard form}, or equivalently, in \emph{systematic form}. Hence, for a given matrix $X\in\Fq^{k \times(n-k)}$, we denote by $\C_X$ the code generated by $(I_k \mid X)$. 
 It is well-known that MDS codes can be characterized by the non-systematic part of their generator matrix in standard form. Concretely, we have the following result.

\begin{theorem}\label{thm:MDSsuperregular}
A linear code $\C_X\subseteq \Fq^n$ is MDS if and only if the matrix $X$ is \emph{superregular}\footnote{A matrix $A \in \F^{r \times t}$ is said to be superregular if all its minors are nonzero} 
\end{theorem}

Let $0<k\leq n$ be two positive integers, and consider the set of polynomials over $\F_q$ of degree strictly less than $k$
 $$ \Fq[x]_{<k}:= \left\{f(x) \in \Fq[x] \mid \deg f<k\right\}.$$

\begin{definition}
Suppose moreover that $n\leq q$, and consider $\alpha_1,\dots, \alpha_n \in \Fq$ pairwise distinct elements, and $b_1,\ldots, b_n \in \Fq^*$. The code
 $$\C= \left\{(b_1f(\alpha_1),b_2f(\alpha_2),\ldots,b_nf(\alpha_n))\mid f\in \Fq[x]_{<k}\right\}$$
is called \emph{Generalized Reed-Solomon (GRS) code} and it is denoted by $\mathrm{GRS}_{n,k}(\bs\alpha, \bs b)$,
where $\bs \alpha=(\alpha_1,\ldots, \alpha_n)$ and $\bs b=(b_1,\ldots, b_n)$.
\end{definition}

It is well-known that GRS codes are MDS and that the canonical generator matrix for a GRS code $\C=\mathrm{GRS}_{n,k}(\bs\alpha, \bs b)$ is given by the  \emph{weighted Vandermonde matrix} that is

\[ \left( \begin{array}{cccc} b_1 & b_2 &\dots & b_n \\  b_1\alpha_1 & b_2\alpha_2 &\dots & b_n\alpha_n \\
 b_1\alpha_1^2 & b_2\alpha_2^2 &\dots & b_n\alpha_n^2 \\
    \vdots& \vdots&&\vdots \\  b_1\alpha_1^{k-1} &  b_2\alpha_2^{k-1} &\dots
    & b_n\alpha_n^{k-1} \end{array}\right)= V_k(\bs\alpha) \mathrm{diag}(\bs b),\]
where $V_k(\bs\alpha)$ is the classical Vandermonde matrix, and $\mathrm{diag}(\bs b)$ denotes the diagonal matrix whose diagonal entries are given by $b_1,\ldots, b_n$.  This generator matrix is  obtained by choosing the set of monomials $\left\{1,x,x^2,\ldots, x^{k-1} \right\}$ as an $\F_q$-basis of $\Fq[x]_{<k}$, and then evaluating each of them in the points $\alpha_1,\ldots, \alpha_n$. This is why we refer to it as the \emph{canonical generator matrix}.

In 1985 Roth and Seroussi \cite{ro85} studied the generator matrix in standard form of a GRS code, giving a complete characterization. The same result was given by D\"ur in \cite{dur1987automorphism}.

\begin{definition}\label{def:cauchy}
Let $r,s$ be  positive integers, $x_1,\ldots,x_r, y_1,\ldots,y_s \in \F_q$,  and $c_1,\ldots, c_r$, $d_1,\ldots,d_s \in \Fq^*$ be elements such that
\begin{itemize}
\item[(a)] $ x_1,\ldots,x_r$ pairwise distinct,
\item[(b)] $ y_1,\ldots,y_{s}$ pairwise distinct,
\item[(c)] $y_i \in \Fq\setminus \{x_1,\ldots, x_r\}$, for $i=1,\ldots, s$.
\end{itemize}
The matrix $C \in \Fq^{r\times s}$ defined by
$$C_{{i,j}}={\frac  {c_{i}d_{j}}{x_{i}-y_{j}}}$$
is called \emph{Generalized Cauchy (GC) matrix}.
\end{definition}

\begin{theorem}\cite[Theorem 1]{ro85}\label{thm:GRSGC}, \cite[Theorem 2]{dur1987automorphism}
 \begin{enumerate}
\item If $\C=\mathrm{GRS}_{n,k}(\bs\alpha, \bs b) $, then $\C=\C_X$, where 
$X \in \Fq^{k \times (n-k)}$ is a GC matrix.
\item If $X \in \Fq^{k \times (n-k)}$ is a GC matrix then the code $\C_X$ 
 is a Generalized Reed-Solomon code.
\end{enumerate}
\end{theorem}

Theorem \ref{thm:GRSGC} gives a correspondence between GRS codes of dimension $k$ and length $n$ over $\F_q$,  and $k\times (n-k)$ GC matrices over $\F_q$. Moreover, in \cite[Lemma 7]{ro89}, a characterization of the GC in terms of its entries was given. We are now going to reformulate this result for our purpose, in order to underline that it gives a way to determine whether a code is a GRS code in terms of its generator matrix in standard form. 

 Let $A \in (\Fq^*)^{r \times s}$ with entries $a_{i,j}$. We denote by $A^{(-1)}$ the $r \times s$ matrix over $\Fq^*$ whose entries are $a_{i,j}^{-1}$. 

\begin{theorem}\cite{ro89}\label{thm:MDSGC}
 Let $X\in \Fq^{k \times (n-k)}$ Then, the code $\C_X$ is a GRS code if and only if 
\begin{enumerate}
 \item[(i)] every entry $x_{i,j}$ is non-zero,
 \item[(ii)] every $2 \times 2$ minor of $X^{(-1)}$ is non-zero, and
\item[(iii)] $\rk(X^{(-1)})=2$.
\end{enumerate}
\end{theorem}

In the following we will see that the analogue of GRS in the rank metric is given by  generalized Gabidulin codes. We will find the same kind of correspondence between them and the rank analogue of GC matrices, obtained by characterizing their generator matrix in standard form. Moreover, we will also find an analogue of Theorem \ref{thm:MDSGC} in that framework.

\section{Rank Metric Codes}

In this section we will give a recap about rank metric codes. In particular, we will only study those that linear over the extension field.
Given a finite field $\F_q$ and an extension field $\Fm$,
recall that $\F_{q^m}$ is isomorphic, as
an $\F_q$-vector space, to $\F_q^m$.
Using this fact, one then easily obtains the isomorphic description of matrices over
the base field $\F_q$ as vectors over the extension field, i.e.\
$\F_q^{m\times n}\cong \F_{q^m}^n$. In this setting, let $g=(g_1,\ldots, g_n)\in \Fm^n$. We define the \emph{$\Fm$-support of $g$ over $\Fq$} the $\Fq$-subspace
$$\rsu_q(g):=\langle g_1,\ldots, g_n\rangle_{\Fq}.$$
Moreover, we denote by $\rk_q(g)=\dim_{\Fq}(\rsu_q(g))$, which is called the \emph{$q$-rank of $g$}.

Unless otherwise specified, whenever we talk about vectors in $\F^n$ over a field $\F$, in this work we will always mean row vectors.
\begin{definition}
  The \emph{rank distance} $d_R$ on $\F_q^{m\times n}$ is defined by
  \[\mathrm{d}_R(X,Y):= \rk(X-Y) , \quad X,Y \in \F_q^{m\times n}. \]
  Analogously, we define the rank distance between two elements
  $\boldsymbol x,\boldsymbol y \in \F_{q^m}^n$ as 
$$\mathrm{d}_r(\bs x, \bs y):=\rk_q(\bs x-\bs y),$$
which corresponds to the rank of the
  difference of the respective matrix representations in
  $\F_q^{m\times n}$.
\end{definition}

In this paper we will focus on $ \F_{q^m}$-linear rank metric codes in
$\F_{q^m}^n$, i.e.\ those codes that form a subspace of $
\F_{q^m}^n$.
\begin{definition}
  An $\F_{q^m}$-\emph{linear rank metric code $\mathcal C$} of length
  $n$ and dimension $k$ is a $k$-dimensional $\Fm$-subspace of $\F_{q^m}^n$
  equipped with the rank distance $\mathrm{d}_r$.
\end{definition}

As in the Hamming metric case, one defines the \emph{minimum rank distance} of $\C$ as
$$\mathrm{d}_r(\C):=\min\left\{\mathrm{d}_r(u,v) \mid u,v\in \C, u\neq v\right\},$$  
and a \emph{generator matrix} $G\in\F_{q^m}^{k\times n}
  $  as a matrix whose row space (over $\Fm$) is $\C$.

The well-known
Singleton bound for codes in the Hamming metric implies  an upper
bound for rank metric codes.

\begin{theorem}\cite[Section~2]{ga85a}\label{thm:singrank}
  Let $\mathcal{C}\subseteq \F_{q^m}^{n}$ be an $\Fm$-linear rank metric code with
  minimum rank distance $d$ of dimension $k$. Then
$$ d\leq n-k+1  .$$
\end{theorem}

\begin{definition}
  A rank metric code meeting the  bound in Theorem \ref{thm:singrank}  is called a \emph{maximum rank
    distance (MRD) code}.
\end{definition}

\begin{lemma}\cite[Lemma 5.3]{ho16}\label{lem:systematic}
  Any $\Fm$-linear MRD code $\C \subseteq \F_{q^m}^n$ of dimension $k$ has a
  generator matrix $G \in \F_{q^m}^{k\times n}$ of the form
  $$G = \left(\begin{array}{c|c}
      I_k \; & X
    \end{array}
  \right).$$ 
  Moreover, all entries in $X$ are from $\F_{q^m} \setminus
  \F_q$.
\end{lemma}

A generator matrix of the form $G=(I_k \mid X)$ is said to be \emph{in standard form} (also called \emph{systematic form}). The matrix $X$ of this representation is the \emph{non-systematic part} of $G$.

Since we are going to deal only with $\Fm$-linear MRD codes, we can denote by $\C_X$ the code generated by $(I_k \mid X)$. In fact, by Lemma \ref{lem:systematic} every MRD code can be represented in a unique way as a code of the form $\C_X$ for some $X\in \F_{q^m}^{k\times (n-k)}$. We will widely use this notation later in this work.

It can be easily shown that a necessary condition for the existence of MRD codes is $n\leq m$. Therefore, in the rest of the paper we will always consider positive integers $k,n,m$ such that $0<k<n\leq m$.

Furthermore, the condition $ n\leq m$ is also sufficient. In \cite{de78, ga85a} a general construction for MRD codes is given, which  has been then generalized in \cite{ks05}. In order to present such a construction we need to introduce a particular class of polynomials.

\begin{definition}
A \emph{linearized polynomial} over $\Fm$ is a polynomial $f(x)\in \Fm[x]/(x^{q^m}-x)$ of the form
$$\sum_{i=0}^{m-1}f_ix^{[i]},$$
where $[i]:=q^i$.
We denote by $\mathcal L_m(\Fm)$ the space of linearized polynomials over $\Fm$.
\end{definition}

Let $\mathcal G_{k,s} \subseteq \mathcal L_m(\Fm)$ be the set defined as
$$ \mathcal G_{k,s}:= \left\{f_0x+f_1x^{[s]}+\ldots+f_{k-1}x^{[s(k-1)]} \mid f_i\in \Fm\right\}.$$

\begin{definition}\label{def:Gab}
Let  $g=(g_1,\ldots, g_n) \in \Fm^n$ be a vector with $\rk_q(g)=n$ and let $s$ be an integer coprime to $m$. Let $\C$ be the rank metric code defined as
$$\C= \left\{(f(g_1),f(g_2),\ldots,f(g_n))\mid f\in \mathcal G_{k,s}\right\}.$$
Then $\C$ is  called \emph{generalized Gabidulin code} of parameter $s$, and it will be denoted by
$$\C=\mathcal G_{k,s}(g).$$
\end{definition}

We denote by $\GL_n(q):=\{A\in \F_q^{n\times n} \mid \rk (A) =n\}$ the
general linear group of degree $n$ over $\F_q$.
Furthermore, given a finite field $\Fq$, we consider the Grassmannian $\Gr(k,\Fq^n)$, that is the set of all $k$-dimensional subspaces of the vector space $\Fq^n$ over $\Fq$. It is well known that its cardinality is given by the Gaussian binomial  $ \binom{n}{k}_{q}$, defined as
$$  \binom{n}{k}_q = \prod_{i=0}^{k-1} \frac{q^n-q^i}{q^k-q^i}=\frac{\prod_{i=0}^{k-1}(q^n-q^i)}{|\GL_k(q)|}.$$
With this notation, for a positive integer $s$ coprime to $m$, we introduce the set $\Gab_q(k,n,m,s)$ as the set of all generalized Gabidulin codes over $\Fm$ of dimension $k$, length $n$ and parameter $s$, i.e.
$$\Gab_q(k,n,m,s):= \left\{ \mathcal U \in \Gr(k,\Fm^n) \mid \mathcal U \mbox{ is a gen. Gabidulin code of parameter } s \right\}.$$

\begin{definition}
For a vector $\bs{v}:=(v_1,\dots, v_n) \in \F_{q^m}^n$ we denote the $k
\times n$ \emph{$s$-Moore matrix} by
\[M_{s,k}(\bs{v}) := \left( \begin{array}{cccc} v_1 & v_2
    &\dots &v_n \\ \theta^s(v_1) & \theta^s(v_2) &\dots &\theta^s(v_n) \\
    \vdots&&&\vdots \\ \theta^{(k-1)s}(v_1)& \theta^{(k-1)s}(v_2) &\dots
    &\theta^{(k-1)s}(v_n) \end{array}\right) .\] 
\end{definition}

At this point, it is straightforward to see that a generator matrix of a generalized Gabidulin code $\mathcal G_{k,s}(g)$ is given by the $k \times n$ $s$-Moore matrix $M_{s,k}(g)$. This generator matrix is said to be \emph{canonical}, since it is obtained by evaluating the  basis of monomials $\{ x, x^{[s]}, x^{[2s]}, \ldots, x^{[(k-1)s]} \}$ of $\mathcal G_{k,s}$ in the points $g_1,\ldots,g_n$. Therefore, the $s$-Moore matrix is the natural rank analogue of a weighted Vandermonde matrix.

Note that for $s=1$, Definition \ref{def:Gab} coincides with the
classical Gabidulin code construction.
 The following theorem was shown
for $s=1$ in \cite[Section 4]{ga85a}, and for general $s$ in
\cite{ks05}.

\begin{theorem}\label{thm:GabisMRD}
 Let $1 \leq k \leq n \leq m$ be integers and let $s$ be another integer coprime to $m$. Moreover, let $g\in \Fm^n$ be such that  $\rk_q(g)=n$. Then, the generalized Gabidulin code $\G_{k,s}(g)\subseteq \F_{q^m}^{n}$ of
  dimension $k$ over $\F_{q^m}$ has minimum rank
  distance 
  $n-k+1$. Thus, generalized Gabidulin codes are $\Fm$-linear MRD codes.
\end{theorem}

The dual code of a code $\mathcal{C}\subseteq \F_{q^{m}}^{n}$ is
defined in the usual way as
\[\mathcal{C}^{\perp} := \{\boldsymbol{u} \in \F_{q^{m}}^{n} \mid
\boldsymbol{u}\boldsymbol{c}^\top=0 \quad \forall \boldsymbol{c}\in
\mathcal{C}\}.  \]

In his seminal paper Gabidulin showed the following two results on
dual codes of MRD and Gabidulin codes. The result was generalized to
$s>1$ later on by Kshevetskiy and Gabidulin. Observe that also Delsarte in \cite{de78} proved a similar result for what concerns the dual of matrix codes with respect to the Delsarte bilinear form.

\begin{proposition}\cite[Sections~2 and 4]{ga85a}\cite[Subsection
  IV.C]{ks05}\label{prop:dual1}
  \begin{enumerate}
  \item Let $\mathcal{C}\subseteq \F_{q^{m}}^{n}$ be an $\Fm$-linear MRD code of
    dimension $k$. Then the dual code $\mathcal{C}^{\perp}\subseteq
    \F_{q^{m}}^{n}$ is an $\Fm$-linear MRD code of dimension $n-k$.
  \item Let $\mathcal{C}\subseteq \F_{q^{m}}^{n}$ be a generalized
    Gabidulin code of dimension $k$ and parameter $s$. Then the dual code
    $\mathcal{C}^{\perp}\subseteq \F_{q^{m}}^{n}$ is a generalized
    Gabidulin code of dimension $n-k$ and parameter $s$.
  \end{enumerate}
\end{proposition}

Given a matrix (resp.\ a vector) $A\in \Fm^{k \times n}$, we
denote by $\theta^s(A)$ the component-wise $q$-Frobenius of $A$ applied $s$ times, i.e. $\theta^s(A)$ is generated by applying $\theta^s$ to every
entry of the matrix (resp.\ the vector) $A$. Analogously, given a code $\C \subseteq \Fm^{k
  \times n}$, we define
$$ \theta^s(\C):=\left\{\theta^s(\mathbf{c})\mid \mathbf{c}\in \C \right\}.$$

Moreover  we consider  the map
$\Phi_s$ defined as
 $$\begin{array}{rcl}
   \Phi_s:\F_{q^m}^{k\times (n-k)} &\longrightarrow & \F_{q^m}^{k\times (n-k)} \\
   X & \longmapsto & \theta^s(X)-X.
 \end{array}
$$
Observe that $\Phi_s$ is the function that maps every entry
$x_{i,j}$ of the matrix $X$ to $\varphi_s(x_{i,j})$.

Here we present some criteria on the generator matrix of a rank metric code, that allow to verify whether the code is MRD or generalized Gabidulin. We will need these results later on. The following criterion was given
in \cite[Proposition 22]{ne18}, and it improves \cite[Corollary 2.12]{ho16}, which in turn is based on a well-known result given in
\cite{ga85a}. First we define the sets

\begin{align*}
    \E_q(k,n)& :=\left\{E \in \F_q^{k\times n} \mid \rk_q(E)=k\right\}, \\
    \mathcal T_q(k,n)&:=\left\{ E \in \E_q(k,n) \mid E \mbox{ is in reduced row echelon form }\right\}.
\end{align*}

\begin{proposition}[new MRD criterion]\cite[Proposition 22]{ne18}\label{prop:newMRDCrit}
  Let $G\in \F_{q^m}^{k\times n}$ be a generator matrix of a
  rank metric code $\mathcal{C}\subseteq \F_{q^m}^n$. Then
  $\mathcal{C}$ is an MRD code if and only if
$$ \rk(EG^\top) =k$$
for all $E\in \mathcal T_q(k,n)$.
\end{proposition}

Furthermore, we need the following criterion for  generalized Gabidulin codes. 

\begin{theorem}[gen.\ Gabidulin criterion]\cite[Lemma 19]{ne18}\label{thm:GabCrit}
  Let $X\in \Fm^{k\times(n-k)}$ such that $\mathcal{C}_X\subseteq \F_{q^m}^n$ is an $\Fm$-linear MRD code. $\mathcal{C}_X$ is a generalized Gabidulin code if and only if
  there exists a positive integer $s$ with $\gcd(s,m)=1$, such that
$$ \rk (\Phi_s(X)) = 1 .$$
\end{theorem}

Theorem \ref{thm:GabCrit} will be one of the most important results on which this work is based. The criterion starts with the assumption that we already know that the code is MRD. However,  in Section \ref{sec:stform} we will derive a new criterion that does not have such assumption and it is definitely easier to verify.

Concerning Gabidulin codes, we can also find the exact number of them. In \cite{be03}, Berger provided the following result.

\begin{proposition}\cite[Theorem~2]{be03}\label{thm:numGab} Let $0<k<n$, and let $g=(g_1,\ldots,g_n)$, $g'=(g_1',\ldots,g_n') \in \Fm^n$ be two vectors such that $\rk_q(g)=\rk_q(g')=n$. Then, for any integer $s$ coprime to $m$,
$\rs(M_{s,k}(g))=\rs(M_{s,k}(g'))$ if and only if $g=\lambda g'$ for some $\lambda \in \Fm^*$.
\end{proposition}

\begin{corollary}\label{cor:numGab}
The number of  $k$-dimensional generalized Gabidulin codes of length $n$ and parameter $s$ over $\Fm$ satisfies
$$|\mathrm{Gab}_{q}(k,n,m,s)| = \prod_{i=1}^{n-1}(q^m-q^i).$$
\end{corollary}

Denote by $\mathrm{Aut}(\F_{q^m})$ the \emph{automorphism group} of
$\F_{q^m}$. It is
well-known that, if $q=p^h$ for a prime $p$, then $\mathrm{Aut}(\F_{q^m})$ is generated by the
\emph{Frobenius map}, which takes an element to its $p$-th
power. Hence, the automorphisms are of the form $x\mapsto x^{p^i}$ for
some $0\leq i < hm$.

The semilinear rank isometries on $\F_{q^m}^{n}$ are induced by the
isometries on $\F_q^{m\times n}$ and are hence well-known, see e.g.\
\cite{be03,mo14,wa96}.
\begin{lemma}\cite[Proposition~2]{mo14}\label{isometries}
  The semilinear $\F_q$-rank isometries on $\F_{q^m}^{n}$ are of the
  form
  \[(\lambda, A, \sigma) \in \left( \F_{q^m}^* \times \GL_n(q) \right)
  \rtimes \mathrm{Aut}(\F_{q^m}) ,\] acting on $ \F_{q^m}^n$ via
  \[(v_1,\dots,v_n) \cdot (\lambda, A, \sigma) = (\sigma(\lambda
  v_1),\dots,\sigma(\lambda v_n)) A .\] In particular, if
  $\mathcal{C}\subseteq \F_{q^m}^n$ is a linear code with minimum rank
  distance $d$, then
  $\mathcal{C}' := \sigma(\lambda \mathcal{C}) A $
  is a linear code with minimum rank distance $d$.
\end{lemma}

As semilinear isometries on $\F_{q^m}^n$ preserve the rank, we get that $\F_q$-linearly independent elements in $\F_{q^m}^n$ remain $\F_q$-linearly independent under the actions of $\left( \F_{q^m}^* \times \GL_n(q) \right)
  \rtimes \mathrm{Aut}(\F_{q^m}) $. 
Moreover,
the $s$-Moore matrix structure is preserved under these actions, which
implies that the class of generalized Gabidulin codes is closed under
the semilinear isometries. Thus, a code is semilinearly isometric to a
generalized Gabidulin code if and only if it is itself a generalized
Gabidulin code.

As a consequence of Lemma \ref{isometries}, we have an interesting result, that  will be useful in the next section.

\begin{corollary}\label{cor:FqGab}
Let $X \in \Fm^{k \times (n-k)}$, and $\tilde{X}=X+B$ for some matrix $B \in \Fq^{k\times (n-k)}$.  Moreover, let $s$ be a positive integer coprime to $m$.
\begin{enumerate}
\item If the code $\C_X$ is MRD, then also $\C_{\tilde{X}}$ is MRD.
\item If the code $\C_X$ is a generalized Gabidulin code of parameter $s$, then also $\C_{\tilde{X}}$ is a generalized Gabidulin code of parameter $s$.
\end{enumerate}
\end{corollary}

\begin{proof} 
\begin{enumerate}
\item Let $G=(I_k\mid X)$, be the generator matrix in standard form for $\C_X$, and let $\widetilde{G}=(I_k \mid \tilde{X})$. 
Then, $\widetilde{G}=GM$ where 
$$M=\begin{pmatrix}
I_k &  B \\
0 & I_{n-k}
\end{pmatrix}\in \GL_n(q).$$
By Lemma \ref{isometries}, $\C_{\tilde{X}}=\C_{X}M$ is MRD.

\item  By Theorem \ref{thm:GabisMRD} the code $C_X$ is MRD, and so it is $C_{\tilde{X}}$ by part (1) of this Corollary. Moreover we have
$$
\Phi_s(\tilde{X})=\Phi_s(X+B) =\Phi_s(X),
$$
and we  conclude using Theorem \ref{thm:GabCrit}.
\end{enumerate}
\end{proof}

We now give an easy improvement of Lemma \ref{lem:systematic}. 

\begin{lemma}\label{lem:nonMRD}
Let $X=(x_{i,j}) \in \Fm^{k\times(n-k)}$. 
\begin{enumerate}
 \item If there exists $i$ such that $\rk_q(1, x_{i,1}, \ldots, x_{i,n-k})<n-k+1$, then $\C_X$ is not MRD.
 
 \item If there exists $j$ such that $\rk_q(1, x_{1,j}, \ldots, x_{k,j})<k+1$, then $\C_X$ is not MRD.
\end{enumerate}

\end{lemma}

\begin{proof}
 \begin{enumerate}
\item Suppose that $1, x_{i,1},\ldots,x_{i,n-k} $ are $\F_q$-linearly dependent for some $i \in \{1,\ldots, k\}$, and consider the non-zero codeword 
$$e_i \left( \begin{array}{c|c} I_k & X \end{array}\right)=(0,\ldots, 0,1,0,\ldots,0, x_{i,1},\ldots,x_{i,n-k}).$$
The rank of this codeword is strictly less than $n-k+1$, and therefore $\C_X$ can not be MRD.

 \item In this case we consider the code $\C_X^\perp$. Since a generator matrix for this code is  $(-X^\top \mid I_{n-k})$, we get that $\C_X^{\perp}$ is permutation equivalent to the code $\C_{-X^\top}$.
 By the first part of this Lemma, we have that $\C_{-X^\top}$ is not MRD and therefore the same holds for $\C_X^\perp$. Hence, by part (1) of Proposition \ref{prop:dual1} we can conclude that $\C_X$ is not MRD.
 \end{enumerate}
\end{proof}

The following result derives from \cite[Corollary 3.3]{ho16} and  it gives conditions for a code $C_X$ to be MRD, based only on the matrix $X$. For this purpose, we first introduce the set of normalized upper triangular matrices as 
$$\mathrm{U}_{r}(q):=\left\{ A \in \Fq^{r \times r} \mid a_{i,j}=0 \mbox{ for all } i>j, a_{i,i}=1 \mbox{ for all } i \right\}.$$

\begin{theorem}\label{thm:MRDsuperregular}
Let $X\in \Fm^{k \times (n-k)}$.  The following are equivalent:
\begin{enumerate}
 \item\label{part1} $C_X$ is MRD. 
\item For every $A\in \GL_{k}(q), B \in \GL_{n-k}(q), C \in \Fq^{k\times (n-k)}$, the matrix $AXB+C$ is superregular.
\item\label{part3} For every $A\in \mathrm{U}_{k}(q), B \in \mathrm{U}_{n-k}(q), C \in \Fq^{k\times (n-k)}$, the matrix $AXB+C$ is superregular.
\end{enumerate}
\end{theorem}

\begin{proof}
\underline{$(2) \Rightarrow (3)$} This is clear, since $\mathrm{U}_{r}(q)\subseteq \GL_{r}(q)$ for any  positive integer $r$.

\underline{$(1) \Rightarrow (2)$} Suppose that $C_X$ is MRD, and let $A\in \GL_{k}(q), B \in \GL_{n-k}(q), C \in \Fq^{k\times (n-k)}$. Then we consider the matrix $\widetilde{G}=(I_k \mid X)M$, where 
$$M:=\npmatrix{A^{-1} & A^{-1}C \\ 0 & B} \in \GL_n(q).$$
Then it is easy to see that $\rs(\widetilde{G})=C_{\widetilde{X}}$, where $\widetilde{X}=AXB+C$. Then the statement follows from Proposition \ref{prop:newMRDCrit} and the characterization of MDS codes given in Theorem \ref{thm:MDSsuperregular}.

\underline{$(3) \Rightarrow (1)$} Suppose that 3 holds. Every matrix $M \in \mathrm{U}_n(q)$ can be written in the form
$$M:=\npmatrix{A & AC \\ 0 & B} \in \GL_n(q),$$
 for some $A\in \mathrm{U}_{k}(q), B \in \mathrm{U}_{n-k}(q), C \in \Fq^{k\times (n-k)}$. Moreover, $\rs((I_k \mid X)M)=C_{\widetilde{X}}$, where $X=A^{-1}XB+C$. Since the map $A\longmapsto A^{-1}$ is a bijection from $\mathrm{U}_k(q)$ into itself, we conclude that $C_X$ is an MRD code using Theorem \ref{thm:MDSsuperregular} and \cite[Corollary 3.3]{ho16}. 
\end{proof}
Theorem \ref{thm:MRDsuperregular}  can be considered as the analogue in the rank metric of Theorem \ref{thm:MDSsuperregular} and can also be found in \cite[Theorem 3.11]{neri2019PhD}. The equivalence between parts \ref{part1} and \ref{part3} has been shown independently in \cite[Theorem 4]{almeida2020}.

\section{Standard Form of Gabidulin Codes}\label{sec:stform}
Analogously to the works of Roth and Seroussi \cite{ro85} and D\"ur \cite{dur1987automorphism} for GRS codes, in this section we characterize the matrices $X\in \Fm^{k \times (n-k)}$ such that the code $\C_X$ is a generalized Gabidulin code, and we  refer to this family of matrices as \emph{$(q,s)$-Cauchy matrices}. In order to do that, we  rely on Theorem \ref{thm:GabCrit} which tells that $\rk(\Phi_s(X))=1$. Therefore, we start with a rank-one matrix $A$ and determine the conditions such that $A$ belongs to the image of the map $\Phi_s$. Finally, we impose that the resulting matrices $X$ with $\Phi_s(X)=A$, are such that the code $\C_X$ is MRD and get the desired characterization.

Furthermore, we also give an analogue of Theorem \ref{thm:MDSGC} for generalized Gabidulin codes. This result represents a new criterion that allows to determine whether a given code in standard form is a generalized Gabidulin code, which is faster than the one given in Theorem \ref{thm:GabCrit}.

As in the whole work, we fix positive integers $0<k<n\leq m$. For every positive integer $s$ with $\gcd(m,s)=1$, we consider the following sets:
\begin{align*}\mathbb G(s)&:=\{X\in \Fm^{k\times(n-k)}\mid \C_X \in \mathrm{Gab}_q(k,n,m,s) \}, \\
  \mathcal R_1^*&:=\left\{A\in (\F_{q^m}^*)^{k\times (n-k)}\,|\, \rk(A)=1\right\}, \\
  \mathcal K
  &:=\left(\ker\left(\mathrm{Tr}_{\F_{q^m}/\F_q}\right)\right)^{k\times(n-k)}.
\end{align*}

\begin{lemma}\label{lem:phi}
 For every integer $s$ coprime to $m$, the following properties hold.
  \begin{enumerate}
\item $\Phi_s(\Fm^{k \times (n-k)})=\mathcal K$.
\item Let $A \in \Fm^{k\times (n-k)}$. If $A \in \mathcal K$ and $X \in \Phi_s^{-1}(\{A\})$, then 
$$\Phi_s^{-1}(\{A\})=\left\{X+B \mid B \in \Fq^{k\times (n-k)}\right\}.$$
In particular, 
$$|\Phi_s^{-1}(\{A\})|=\begin{cases}
    0 &  \mbox{ if } A\notin \mathcal K \\
   q^{k(n-k)} & \mbox{ if } A\in \mathcal K.
  \end{cases}
$$
\item $\Phi_s(\mathbb G(s))\subseteq \mathcal R_1^*\cap \mathcal K$, or, equivalently, $\mathbb G(s) \subseteq \Phi_s^{-1}(\mathcal R_1^* \cap \mathcal K)$.
\item Let $A \in \mathcal R_1^* \cap \mathcal K$ and $X \in \Phi_s^{-1}(\{A\})$. If $X \in \mathbb G(s)$ then the whole preimage of $\{A\}$ is contained in $\mathbb G(s)$, i.e.
$$\Phi_s^{-1}(\{A\}) \subseteq \mathbb G(s).$$

\end{enumerate}
\end{lemma}

\begin{proof}
  \begin{enumerate}
  \item Since $\Phi_s$ is the function that maps every entry
    $x_{i,j}$ of the matrix $X$ to $\varphi_s(x_{i,j})$, we have that $A
    \in \Phi_s(\Fm^{k \times (n-k)})$ if and only if every entry $a_{i,j}$ of
    $A$ belongs to $\mathrm{Im}(\varphi_s)$. By part (5) of Lemma \ref{lem:trace}  this is true if and only if every $a_{i,j}$ belongs to
    $\ker\left(\mathrm{Tr}_{\F_{q^m}/\F_q}\right)$.
  
  \item If $A\notin \mathcal K$, then, by part (1) of this Lemma, this means that
    $\Phi_s^{-1}(A)=\emptyset$. Otherwise, again by part (1),
    $\Phi_s^{-1}(A)\neq\emptyset$. In this case every entry
    $a_{i,j}$ belongs to $\mathrm{Im}(\varphi_s)$, and by part (2) of Lemma \ref{lem:preimage},
$$\varphi_s^{-1}(\{a_{i,j}\})=\left\{x_{i,j}+\lambda \mid \lambda \in \Fq \right\}$$
for some $x \in \Fm$. Since this holds for every entry, we get the desired result.

  \item Let $X \in \mathbb G(s)$. By Theorem \ref{thm:GabCrit}, $\Phi_s(X)$ has rank equal to $1$. Moreover, by Lemma \ref{lem:systematic}, all the entries of $\Phi_s(X)$ are in $\Fm^*$. Finally, by part (1) of this Lemma, we have $\Phi_s(X) \in \mathcal K$ and this concludes the proof.
\item It directly follows from part (2) of this Lemma and part (2) of Corollary \ref{cor:FqGab}.

 \end{enumerate}

\end{proof}

As a consequence of part (4) of Lemma \ref{lem:phi}, given a matrix $X\in \Fm^{k \times (n-k)}$, we have that the property of $\C_X$ being Gabidulin only depends on the image $\Phi_s(X)$. It is now crucial to investigate   the matrices that belong  to the image of the map $\Phi_s$, and,  by part (3) of Lemma \ref{lem:phi}, in particular  $\mathcal R_1^* \cap \mathcal K$.

By definition, every element in  $\mathcal R_1^* \cap \mathcal K$ has rank one, and it is well-known that every rank-one matrix can be written as the product of a non-zero column vector by a non-zero row vector. Moreover, for a fixed rank-one matrix over $\Fm$, there are exactly $q^m-1$ different parametrizations of this form. 

The following result is straightforward and directly follows from the considerations above and the definitions of $\mathcal R_1^*$ and $ \mathcal K$.

\begin{lemma}\label{lem:R1}
 The set  $\mathcal R_1^* \cap \mathcal K$ can be written in the following way
   \begin{align*}
\mathcal R_1^* \cap \mathcal K&=\left\{\bs\alpha^\top\bs\beta \mid \bs\alpha \in \Fm^k, \bs\beta\in\Fm^{n-k}, \alpha_i\beta_j \in \ker(\Tr) \mbox{ for all } i,j\right\} \\
&=\left\{\bs\alpha^\top\bs\beta \mid \bs\alpha \in \Fm^k, \bs\beta\in\Fm^{n-k}, \beta_j \in \rsu_q(\bs\alpha)^{\times} \mbox{ for all } j\right\}\\
&= \left\{\bs\alpha^\top\bs\beta \mid \bs\alpha \in \Fm^k, \bs\beta\in\Fm^{n-k}, \rsu_q(\bs\beta) \rsu_q(\bs\alpha)^{\times} \right\}
\end{align*}
Moreover, every element in $\mathcal R_1^*\cap \mathcal K$ has $q^m-1$ distinct representations of this form.
\end{lemma}

This result gives a convenient way to represent $\mathcal R_1^*\cap \mathcal K$ using the set

$$V_{k,n}:=\left\{(\bs\alpha,\bs\beta)\in \Fm^k \times \Fm^{n-k} \mid  \rsu_q(\bs\beta) \rsu_q(\bs\alpha)^{\times}  \right\}.$$
Notice that, since we have $q^m-1$ distinct representations for a matrix in $\mathcal R_1^*\cap \mathcal K$ and the entries are all non-zero, we can always choose the representation with $\beta_1=1$.

At this point, given $(\bs\alpha, \bs\beta) \in V_{k,n}$ and a matrix $X \in \Phi_s^{-1}(\{\bs\alpha^\top\bs\beta\})$, we have, by Theorem \ref{thm:GabCrit} and by the definition of $\mathbb G(s)$, that $\C_X$ is MRD if and only if $X \in \mathbb G(s)$, i.e. if and only if $\C_X$ is a generalized Gabidulin code of parameter $s$.

\begin{lemma}\label{lem:alphabetaindep}
 Let $(\bs\alpha, \bs\beta) \in V_{k,n}$, where $\bs\alpha=(\alpha_1,\ldots, \alpha_k) $ and $\bs\beta=(\beta_1,\ldots, \beta_{n-k})$ and let
 $$X\in {\Phi}_s^{-1}(\{\bs\alpha^\top\bs\beta\}).$$
 \begin{enumerate}
 \item If $\rk_q(\bs\alpha)<k$, then $X \notin \mathbb G(s)$, i.e. $\C_X$ is not MRD.
 \item If $\rk_q(\bs\beta)<n-k$, then $X \notin \mathbb G(s)$, i.e. $\C_X$ is not MRD.
 \end{enumerate}
\end{lemma}

\begin{proof}
\begin{enumerate}
\item The entries of the first column of $\bs\alpha^\top\bs\beta$ are $\alpha_1\beta_1,\ldots,\alpha_k\beta_1$ that by hypothesis are $\F_q$-linearly dependent. By Lemma \ref{lem:preindip} this  means that  the entries of the first column of $X$ together with the element $1$, are $\F_q$-linearly dependent. At this point we conclude by Lemma \ref{lem:nonMRD}.

\item The entries of the first row of  $\bs\alpha^\top\bs\beta$ are $\alpha_1\beta_1\ldots,\alpha_1\beta_{n-k}$ that by hypothesis are $\F_q$-linearly dependent. Then we conclude again using Lemma \ref{lem:preindip} and Lemma \ref{lem:nonMRD}.
\end{enumerate}
\end{proof}

Finally, we can state our desired result. 

\begin{theorem}[Standard form of Gabidulin codes]\label{thm:stform}
 Suppose $X \in \Fm^{k\times(n-k)}$ is a matrix such that $\C_X \in \mathrm{Gab}_q(k,n,m,s)$. Then $X \in {\Phi}_s^{-1}(\{\bs\alpha^\top\bs\beta\})$ for some $\bs\alpha \in \Fm^k, \bs\beta \in \Fm^{n-k}$ such that
 \begin{itemize}
  \item[(a)] $\rk_q(\bs\alpha)=k$,
  \item[(b)] $\rk_q(\bs\beta)=n-k$,
  \item[(c)] $\rsu_q(\bs\beta) \subseteq \rsu_q(\bs\alpha)^{\times}$.
 \end{itemize}

Moreover, if $\bs\alpha \in \Fm^k, \bs\beta \in \Fm^{n-k}$ satisfy properties (a), (b), (c)  and $X \in {\Phi}_s^{-1}(\{\bs\alpha^\top\bs\beta\})$, then $\C_X\in \mathbb  \mathrm{Gab}_q(k,n,m,s)$.
\end{theorem}

\begin{proof}
Let $\C_X$ be a Gabidulin code. We have that ${\Phi}_s(X)$ is of the form $\bs\alpha^\top\bs\beta$ for some $\bs\alpha, \bs\beta$ by part (3) of Lemma \ref{lem:phi} and by Lemma \ref{lem:R1}. Moreover, part (c) follows from the fact that if $\C_X$ is a Gabidulin code, then all the entries of ${\Phi}_s(X)$ belong to $\ker(\Tr)$. Finally part (a) and (b) follow from Lemma \ref{lem:alphabetaindep}.

On the other hand, we can count the number of matrices $X\in{\Phi}_s^{-1}(\{\bs\alpha^\top\bs\beta\})$ for  $\bs\alpha \in \Fm^k, \bs\beta \in \Fm^{n-k}$ satisfying properties (a), (b), (c). For $\bs\alpha$ we have $\prod_{i=0}^{k-1}(q^{m}-q^i)$ possible choices, while for $\bs\beta$ we have $\prod_{i=0}^{n-k-1}(q^{m-k}-q^i)$ choices. Moreover we need to divide by $q^m-1$ since, by Lemma \ref{lem:R1}, we have $q^m-1$ choices of $(\bs\alpha,\bs\beta)$ that gives the same matrix $\bs\alpha^\top\bs\beta$.
Since for every $\bs\alpha^\top\bs\beta \in \mathcal R_1^*\cap\mathcal K$ we have, by part (2) of Lemma \ref{lem:phi}, $q^{k(n-k)}$ many matrices in the preimage under the map ${\Phi}_s$, we finally obtain 
\begin{align*}
& \frac{q^{k(n-k)}}{q^m-1}\prod_{i=0}^{k-1}(q^{m}-q^i)\prod_{i=0}^{n-k-1}(q^{m-k}-q^i) \\
=&\prod_{i=1}^{k-1}(q^m-q^i)\prod_{i=0}^{n-k-1}(q^{m}-q^{i+k})\\
= &\prod_{i=1}^{n-1}(q^m-q^i).
\end{align*}
By Corollary \ref{cor:numGab}, this number is equal to the number of distinct Gabidulin codes. Therefore, by a counting argument, it follows that conditions (a), (b), (c) are also sufficient.
\end{proof}

Theorem \ref{thm:stform}  gives a characterization of the generator matrix in standard form of a generalized Gabidulin code.
 In \cite{ro85, dur1987automorphism}, it was shown that there is a one-to-one correspondence between generalized Reed-Solomon (GRS)  codes and generalized Cauchy (GC) matrices. In that paper, it is shown that a code in the Hamming metric whose generator matrix in standard form is $(I_k \mid X)$ is a GRS code if and only if $X$ is a GC matrix. Since generalized Gabidulin codes are the analogue of GRS codes for the rank metric, it becomes natural to give the definition of a $q$-analogue of Cauchy matrices according to Theorem \ref{thm:stform}.

Let  $\gamma \in \Fm$ such that $\Tr(\gamma)\neq 0$ and $s$ an integer coprime to $m$. We define the function $\pi_s$ as
\begin{equation}
\begin{array}{ >{\displaystyle}r >{{}}c<{{}}  >{\displaystyle}l } 
\pi_s : \Fm & \longrightarrow & \Fm \\
\alpha & \longmapsto & 
\frac{-1}{\Tr(\gamma)} \sum_{i=0}^{m-2}\left(\sigma^{i+1}(\gamma)\sum_{j=0}^i(\sigma^j(\alpha))\right).
\end{array}
\end{equation}
where $\sigma:=\theta^s$. Recall that, by Lemma \ref{lem:preimage}, for $\alpha \in \ker(\Tr)$, $\pi_s(\alpha)$ gives one of the elements in the preimage of $\varphi_s$, i.e.  $\varphi_s(\pi_s(\alpha))=\alpha$ and
 $\pi_s(\varphi_s(\alpha))=\alpha+\lambda$ for some $\lambda \in \Fq$. Moreover, every element in $\varphi_s^{-1}(\{\alpha\})$ is of the form $\pi_s(\alpha)+\lambda$. 

\begin{definition}\label{def:rankcauchy}
 Let $\bs\alpha \in \Fm^t, \bs\beta \in \Fm^{r}$ such that 
 \begin{itemize}
  \item[(A)] $\rk_q(\bs\alpha)=t$,
  \item[(B)] $\rk_q(\bs\beta)=r$,
  \item[(C)] $\rsu_q(\bs\beta) \subseteq \rsu_q(\bs\alpha)^{\times}$.
 \end{itemize}
 Moreover, let  $s$ be an integer coprime to $m$ and $B \in \F_q^{t\times r}$. A $t\times r$ \emph{$(q,s)$-Cauchy matrix} $C_{(q,s)}(\bs\alpha,\bs\beta,B)$ of parameter $s$ is a matrix of the form
$$C_{(q,s)}(\bs\alpha,\bs\beta,B)=\begin{pmatrix}  \pi_s(\alpha_1\beta_1) & \pi_s(\alpha_1\beta_2) & \cdots & \pi_s(\alpha_1\beta_{r})\\
 \pi_s(\alpha_2\beta_1) & \pi_s(\alpha_2\beta_2) &\cdots & \pi_s(\alpha_2\beta_{r})\\
\vdots & \vdots &  & \vdots \\
 \pi_s(\alpha_t\beta_1) & \pi_s(\alpha_t\beta_2) &\cdots & \pi_s(\alpha_t\beta_{r})\\
\end{pmatrix}+B.$$
When $s=1$ we will simply call it \emph{$q$-Cauchy matrix}.
\end{definition}

\begin{remark}
 Definition \ref{def:rankcauchy} directly arises from the characterization of the generator matrix in standard form of a Gabidulin code. However, one can see that $(q,s)$-Cauchy matrices introduced in this work are the $q$-analogue of GC matrices. Indeed,  conditions (A), (B) and (C) represent  the $q$-analogues of conditions (a), (b) and (c) of Definition \ref{def:cauchy}.
 \end{remark}

With this definition, we can reformulate Theorem \ref{thm:stform} in the following way, that puts emphasis on the correspondence between generalized Gabidulin codes and $(q,s)$-Cauchy matrices

\begin{thmbis}\label{thm:stformbis}
 Let $X \in \Fm^{k \times(n-k)}$ and let $s$ be a positive integer coprime to $m$. Then,  $\C_X \in \Gab_q(k,n,m,s)$  if and only if the matrix $X$ is a $(q,s)$-Cauchy matrix.
\end{thmbis}

From Theorem \ref{thm:stform} we have an immediate consequence, that relates $(q,s)$-Cauchy matrices with Moore matrices.
\begin{corollary}\label{cor:inverseMoore}
 Let $0<k<n\leq m$ be positive integers and $s$ be another integer coprime to $m$. Let  $g\in \Fm^n$ be  such that $\rk_q(g)=n$. Then the matrix
$$ M_{k,s}(g_1,\ldots,g_k)^{-1} M_{k,s}(g_{k+1},\ldots,g_n) $$
is a $(q,s)$-Cauchy matrix in $\Fm^{k \times (n-k)}$.

Moreover, if $R \in \Fm^{t \times r}$ is a $(q,s)$-Cauchy matrix, then there exists $g=(g_1,\ldots,g_{t+r}) \in \Fm^{t+r}$  with $\rk_q(g)=t+r$  such that
$$R= M_{t,s}(g_1,\ldots,g_t)^{-1} M_{t,s}(g_{t+1},\ldots,g_{t+r}).$$
\end{corollary}

Now, we want to determine the basis of the linearized polynomial space $\mathcal G_{k,s}$ that corresponds to the generator matrix in standard form.
In order to do that, we introduce the following notion. 
\begin{definition}
 Let $h=(h_1,\ldots, h_\ell) \in \Fm^{\ell}$ be a vector such that $\rk_q(h)=\ell$, and let $s$ be an integer coprime to $m$. We define the polynomial $p_{h,s}$ associated to $h$ as
$$p_{h,s}(x)=\det(M_{\ell+1,s}(h_1,\ldots, h_\ell, x)).$$
\end{definition}

Obviously $p_{h,s}(x)$ is a linearized polynomial and in particular it belongs to $\mathcal G_{\ell+1,s}$. Observe that, by the properties of  $s$-Moore matrices, it can be deduced that the set of roots of $p_{h,s}(x)$ in $\Fm$ is equal to the $\Fq$-subspace $\rsu_q(h)$. Moreover, If $h,h'\in \Fm^{\ell}$ are two vectors such that $\rk_q(h)=\rk_q(h')=\ell$ and $\rsu_q(h)=\rsu_q(h')$, then
$$p_{h,s}(x)=\det(E)p_{h',s}(x),$$
where $E\in \Fq^{\ell \times\ell}$ is the change-of-basis matrix from $\{h_1',\ldots,h_{\ell}'\}$ to $\{h_1,\ldots, h_{\ell}\}$. Note that the  polynomial associated to $h$ is a scalar multiple of a particular polynomial that is known as subspace polynomial or annihilator polynomial of the subspace $\rsu_q(h)$
 (see \cite[Chapter 3, Section 4]{li94} for more details).

\begin{remark}\label{rem:systpoly}
 Let $\C=\mathcal G_{k,s}(g_1,\ldots, g_n)$ be a generalized Gabidulin code of parameter $s$. Consider the vectors
$$g^{(i)}:=( g_1,\ldots,g_{i-1},g_{i+1},\ldots, g_k)\in \Fm^{k-1} \quad \mbox{ for } i=1,\ldots, k,$$
and define the polynomials
$$f_i(x):=p_{g^{(i)},s}(g_i)^{-1}p_{g^{(i)},s}(x) \quad \mbox{ for } i=1,\ldots, k.$$
It follows from the definition that for every $i,j \in \{1,\ldots, k\}$, we have $f_i(x)\in \mathcal G_{k,s}$ and
$$f_i(g_j)=\delta_{i,j}=\begin{cases} 1 & \mbox{ if } i=j \\
 0 & \mbox{ if } i\neq j.
\end{cases}$$
Therefore the generator matrix in standard form for  the generalized Gabidulin code $\C$ is obtained evaluating the basis
$\left\{f_1(x), \ldots, f_k(x) \right\}$ of $\mathcal G_{k,s} $
in the vector $g=(g_1,\ldots, g_n)$.
\end{remark}

\subsection{A new criterion for generalized Gabidulin codes}
The following result represents the  analogue of Theorem \ref{thm:MDSGC} for the rank metric, and its proof directly follows from Theorem \ref{thm:stform}.

\begin{theorem}[New Gabidulin Criterion I]\label{thm:GabCritNew}
Let $X \in \Fm^{k \times (n-k)}$ and let $s$ be an integer coprime to $m$. Then, $\C_X\in \Gab_q(k,n,m,s)$  if and only if
\begin{itemize}
\item[(i)] the first row of the matrix $\Phi_s(X)$ has $q$-rank $n-k$
\item[(ii)] the first column of the matrix $\Phi_s(X)$ has $q$-rank $k$,
\item[(iii)] $\rk(\Phi_s(X))=1$.
\end{itemize}
\end{theorem}

This theorem can be reformulated also in the following way.

\begin{thmbis}[New Gabidulin Criterion II]\label{thm:GabCritNewII}
Let $$X =(x_{i,j})_{\substack{1 \leq i \leq k \\ 1 \leq j \leq n-k }} \in \Fm^{k \times (n-k)}$$ and let $s$ be an integer coprime to $m$. Then, $\C_X\in \Gab_q(k,n,m,s)$  if and only if
\begin{itemize}
\item[(i')] $\rk_q(1,x_{1,1}, \ldots, x_{1,n-k})=n-k+1$,
\item[(ii')] $\rk_q(1,x_{1,1}, \ldots, x_{k,1})=k+1$,
\item[(iii)] $\rk(\Phi_s(X))=1$.
\end{itemize}
\end{thmbis}

\begin{proof}
 By Lemma \ref{lem:preindip} we have that conditions (i') and (ii') are equivalent to conditions (i) and (ii).  This means that the statement is equivalent to Theorem \ref{thm:GabCritNew}. 
\end{proof}

In addition to representing a natural analogue of Theorem \ref{thm:MDSGC} for the rank metric framework, Theorem \ref{thm:GabCritNew} also gives a new criterion to recognize whether a given code in standard form is a generalized Gabidulin code. Observe that, contrary to Theorem \ref{thm:GabCrit}, that gives a criterion subject to a previous verification that the code is MRD, this result is independent on this assumption, and it could be verified more easily. Indeed, according to Proposition \ref{prop:newMRDCrit}, checking whether a code is MRD requires the computation of $\binom{n}{k}_q=\mathcal O(q^{k(n-k)})$ matrix products and ranks, while this new criterion only requires to check the linear independence of two sets of elements and the computation of the rank of one matrix.

More generally, suppose we have an $\Fm$-linear rank metric code $\C\subseteq \Fm^n$ of dimension $k$ given by one of its generator matrices $G\in \Fm^{k\times n}$, and an integer $s$ coprime to $m$. We can check whether $\C$ is a generalized Gabidulin code of parameter $s$ with the following algorithm. First we compute the reduced row echelon form of $G$. If it is not of the form $(I_k \mid X)$, then by Lemma \ref{lem:systematic}, $\C$ is not MRD and hence it is not a generalized Gabidulin code for any parameter $s$. Hence, suppose we get a matrix of the form $(I_k \mid X)$. We can  use Theorem \ref{thm:GabCritNew}, computing the matrix $\Phi_s(X)$ and  its rank, and then verifying the linear independence of the elements in first row and in the first column. It is easy to see that the computational cost of this algorithm is given by the cost of computing the reduced row echelon form of $G$, that can be done via Gaussian elimination, or with faster algorithms. 
Therefore we have just provided a procedure that verifies if a given code is a Gabidulin code with 
$\mathcal O(m\cdot F(k,n))$ operations over $\Fm$, where $F(k,n)$ represents the computational cost of computing the reduced row echelon form of a $k \times n$ matrix. 
Observe that new criteria for checking whether a given rank-metric code is a generalized Gabidulin code were given in \cite[Theorem 6.5]{neri2019equivalence}. Although these criteria do not require to check the MRD condition, it is easy to see that the procedure described above is still faster.

\begin{example}
Let $q=3$, $k=3$ and $n=m=6$. Consider the finite field $\F_{3^6}=\F_3(a)$, where $a$ is a primitive element that satisfies the relation $a^6+2a^4+a^2+2a+2=0$. Consider the $\F_{3^6}$-linear code $\C\subseteq \F_{3^6}^6$ with generator matrix 
$$G=\begin{pmatrix}a^{2} & a^{54} & a^{591} & a^{277} & a^{160} & a^{634}  \\ 
a^{67} & a^{701} & a^{443} & a^{45} & a^{486} & a^{209}  \\
a^{320} & a^{199} & a^{650} & a^{361} & a^{701} & a^{562}  \end{pmatrix}.$$
We put $G$ in reduced row echelon form, and obtain the matrix $(I_3\mid X)$ with
$$X=\begin{pmatrix} a^{180} & a^{373} & a^{714} \\ 
a^{14} & a^{588} & a^{561} \\
a^{370} & a^{702} & a^{442} \end{pmatrix}.$$
For $s=1$ we consider the map 
$$ \begin{array}{ >{\displaystyle}r >{{}}c<{{}}  >{\displaystyle}l } 
\pi_1 : \F_{3^6} & \longrightarrow & \F_{3^6} \\
z & \longmapsto & 
 \sum_{i=0}^{4}\left(\gamma^{3^{i+1}}\sum_{j=0}^iz^{3^j}\right),\end{array}$$
 with $\gamma=a^2$. Then, we compute the matrix
$$\Phi_1(X)=\begin{pmatrix}  a^{72} & a^{226} & a^{406} \\ 
a^{98} & a^{252} & a^{432} \\
a^{144} & a^{298} & a^{478}  \end{pmatrix} $$
and observe that it has rank one. Moreover, the element of the first row of $\Phi_1(X)$ are linearly independent over $\F_3$ and the same holds for the elements of the first column. Thus, by Theorem \ref{thm:GabCritNew}, $\C$ is a generalized Gabidulin code of parameter $1$, that is a classical Gabidulin code.
\end{example}

\subsection{Recovering the parameters of the code from the $(q,s)$-Cauchy matrix}

In order to complete the picture on the correspondence between generalized Gabidulin codes and $(q,s)$-Cauchy matrices, we need to find the relations between the vector of points $g=(g_1,\ldots,g_n)$ in which the set $\mathcal G_{k,s}$ is evaluated, and the corresponding vectors $\bs\alpha=(\alpha_1,\ldots,\alpha_k) \in \Fm^k,\bs\beta=(\beta_1,\ldots,\beta_{n-k})\in \Fm^{n-k}$ and matrix $B\in \F_q^{k\times(n-k)}$ that define the $(q,s)$-Cauchy matrix.  Observe that, by Lemma \ref{lem:R1} and Proposition \ref{thm:numGab}, we can always suppose $\beta_1=g_1=1$. In the rest of this section we will always use this assumption.

 As a preliminary result, we prove that knowing the entries of a $(q,s)$-Cauchy matrix is equivalent to knowing its defining parameters $\bs \alpha, \bs \beta$ and $B$.
If one knows the latter, then it is trivial that the entries of the $(q,s)$-Cauchy matrix can be easily computed. For the other way around we have the following result.

\begin{proposition}\label{prop:equivGRC}
Let $X\in \Fm^{t\times r}$ be a $(q,s)$-Cauchy matrix. Then it is possible to recover the parameters $\bs\alpha\in \Fm^t,\bs\beta \in \Fm^{r}$ and $B\in \Fq^{t\times r}$ from the entries of $X$.
\end{proposition}

\begin{proof}
 It follows from the definition of $\pi_s$ and from Lemma \ref{lem:preimage} that $\varphi_s(x_{i,1})=\alpha_i$ (since $\beta_1=1$), and $\varphi_s(x_{i,j})=\alpha_i\beta_j$ for $j=2,\dots,r$. From that, we can recover $\bs\alpha$ and $\bs\beta$.  Finally, the matrix $B$ can be easily obtained, since 
$$B=X-\begin{pmatrix} \pi_s(\alpha_1\beta_1) & \pi_s(\alpha_1\beta_2) & \cdots & \pi_s(\alpha_1\beta_{r})\\
\pi_s(\alpha_2\beta_1) &  \pi_s(\alpha_2\beta_2) &\cdots & \pi_s(\alpha_2\beta_{r})\\
\vdots & \vdots &  & \vdots \\
\pi_s(\alpha_t\beta_1) &  \pi_s(\alpha_1\beta_2) & \cdots & \pi_s(\alpha_t\beta_{r})\\
\end{pmatrix}.$$
\end{proof}

Suppose we have a generalized Gabidulin code $\C=\mathcal G_{k,s}(g_1,\ldots, g_n)$. Then we can efficiently obtain the corresponding $(q,s)$-Cauchy matrix by computing the reduced row echelon form of the $s$-Moore matrix $M_{k,s}(g_1,\ldots, g_n)$. The cost of this reduction is $\mathcal O(F(k,n))$ field operations over the finite field $\Fm$, where $F(k,n)$ is the cost of computing the reduced row echelon form of a $k \times n$ matrix. If we want a more explicit way to do it (but less efficient), then we can compute the basis  $\{f_1(x),\ldots, f_k(x)\}$ of $\mathcal G_{k,s}$ as described in Remark \ref{rem:systpoly}, and  evaluate it in the vector $g=(g_1,\ldots, g_n)$. In order to recover the parameters  $\bs\alpha\in \Fm^k,\bs\beta \in \Fm^{n-k}$ and $B\in \Fq^{k\times (n-k)}$, one can use Proposition \ref{prop:equivGRC}.

On the other hand, we have that the two sets of parameters that we want to put in relation, are connected by Corollary \ref{cor:inverseMoore} as follows:

$$M_{k,s}(g_1,\ldots,g_k)C_{(q,s)}(\bs\alpha,\bs\beta,B)=M_{k,s}(g_{k+1},\ldots,g_n).$$

From this matrix equation we can deduce how to get the vector $g$ from $\bs\alpha, \bs\beta$ and $B$. Let $\sigma:=\theta^s$, where $\theta$ is the $q$-Frobenius automorphism of $\Fm$. Since $\beta_1=1$, from the first column of the matrix product  we get 
\begin{equation}\label{eq:first}
\sum_{j=1}^kg_j(\pi_s(\alpha_j)+b_{j,1})=g_{k+1}
\end{equation}
and, in general, for $\ell=0,\ldots,k-1$,
\begin{equation}\label{eq:general}
\sum_{j=1}^k\sigma^\ell(g_j)(\pi_s(\alpha_j)+b_{j,1})=\sigma^\ell(g_{k+1}),
\end{equation}
where we have set $\sigma=\theta^s$.
If we apply $\sigma$ to equation (\ref{eq:general}) for $\ell-1$ and we subtract (\ref{eq:general}) to it, we get the set of equations
\begin{align}\label{eq:supgen}
0&=\sum_{j=1}^k\sigma^\ell(g_j)(\sigma(\pi_s(\alpha_j)+b_{j,1})-(\pi_s(\alpha_j)+b_{j,1}))\nonumber \\
 &=\sum_{j=1}^k\sigma^\ell(g_j)\alpha_j,
\end{align}
for every $\ell=1,\ldots, k-1$, where the last identity follows from part (2) of Lemma \ref{lem:preimage}.

We can repeat this process with any other column of the matrix product, and we get, for $i=2,\ldots,n-k$,
\begin{equation}\label{eq:firstgen}
\sum_{j=1}^kg_j(\pi_s(\alpha_j\beta_i)+b_{j,i})=g_{k+i}
\end{equation}
and
$$
0=\sum_{j=1}^k\sigma^\ell(g_j)\alpha_j\beta_i.
$$
However, this set of equations is the same as (\ref{eq:supgen}), therefore we do not consider it. Now, we can show that equations (\ref{eq:first}), (\ref{eq:supgen}) and (\ref{eq:firstgen}) are exactly what we need for our purpose.

By Proposition \ref{prop:equivGRC}, we can recover the vectors $\bs\alpha$ and $\bs\beta$ and the matrix $B$ from $X$. Moreover, applying $\sigma^{-\ell}$ to every equation in (\ref{eq:supgen}), we get a  linear system 
\begin{equation}\label{eq:linsyst}\left(\begin{array}{cccc}
\sigma^{-1}(\alpha_2) & \sigma^{-1}(\alpha_3) & \cdots & \sigma^{-1}(\alpha_k)  \\
\sigma^{-2}(\alpha_2) & \sigma^{-2}(\alpha_3) & \cdots & \sigma^{-2}(\alpha_k)\\
\vdots & \vdots & & \vdots   \\
\sigma^{-k+1}(\alpha_2) & \sigma^{-k+1}(\alpha_3) & \cdots & \sigma^{-k+1}(\alpha_k) \\
\end{array}\right) \left( \begin{array}{c} g_2 \\ \vdots \\ g_k \end{array} \right)=\left( \begin{array}{c}\sigma^{-1}(\alpha_1) \\ \sigma^{-2}(\alpha_1) \\ \vdots \\ \sigma^{-k+1}(\alpha_1)\end{array} \right)\end{equation}
with $g_2,\ldots,g_k$ unknowns.
The matrix defining the linear system (\ref{eq:linsyst}) is  a $(k-1)\times(k-1)$ matrix with coefficients in $\Fm$. In particular, this matrix is equal to the Moore matrix $M_{k-1,-s}(\sigma^{-1}(\alpha_2),\ldots,\sigma^{-1}(\alpha_k))$, and since $\alpha_2,\ldots,\alpha_k$ are $\F_q$-linearly independent it has full rank. The unique solution of this linear system allows to compute $g_2,\ldots, g_k$, and for computing $g_{k+1},\ldots, g_n$ one can use (\ref{eq:first}) and (\ref{eq:firstgen}).

\section{Gabidulin codes in Hankel and Toeplitz form}

In this section we use the characterization of the generator matrix in standard form for a generalized Gabidulin code given in Section \ref{sec:stform} for the construction of particular subclasses of these codes. Indeed, we will prove  that there exist generalized Gabidulin codes $\C_X$ such that $X$ is a Hankel matrix or a Toeplitz matrix.

For our purpose, we first need a technical result.

\begin{lemma}\label{lem:primtrace}
 Let $\gamma\in \Fm$ be a primitive element, i.e. such that $\Fm^*=\langle \gamma \rangle$. Then there exists $\ell \in \mathbb N$ such that 
$$\Tr(\gamma^\ell)=\Tr(\gamma^{\ell+1})=\ldots=\Tr(\gamma^{\ell+m-2})=0.$$
\end{lemma}

\begin{proof}
 Since $\gamma$ is a primitive element, then $\Fm=\F_q(\gamma)$ and $1,\gamma,\gamma^2,\ldots,\gamma^{m-1}$ is an $\F_q$-basis of $\Fm$. Consider the $\Fq$-linear map  $L\in \Hom_{\Fq}(\Fm,\Fq)$ defined as

$$L(\gamma^i)=\begin{cases}
0 & \mbox{ for } 0\leq i \leq m-2 \\
1 & \mbox{ for } i=m-1.
\end{cases}$$
$L$ is a non zero element in $\Hom_{\F_q}(\Fm, \F_q)$, and by Theorem \ref{thm:dualisom},  there exists $\beta \in \Fm^*$ such that $L=T_{\beta}$. At this point, since $\gamma$ is a primitive element,  there exists $\ell \in \mathbb N$ such that $\beta=\gamma^{\ell}$. In this way, we have that for all $i=0,\ldots,m-2$,
$$\Tr(\gamma^{\ell+i})=T_{\gamma^{\ell}}(\gamma^i)=L(\gamma^i)=0$$
and this concludes the proof.
\end{proof}

\begin{definition}
An $r\times s$ matrix $A=(A_{i,j})$ over a field $\F$ is called \emph{Toeplitz matrix} if there exist a vector $a=(a_{1-r},a_{2-r},\ldots, a_{s-1})\in \F^{r+s-1}$ such that
$$A_{i,j}=a_{j-i}.$$

An $r\times s$ matrix $A=(A_{i,j})$ over a field $\F$ is called \emph{Hankel matrix} if there exist a vector $a=(a_{0},a_{1},\ldots, a_{r+s-2})\in \F^{r+s-1}$ such that
$$A_{i,j}=a_{i+j-2}.$$
\end{definition}

A special kind of square Toeplitz matrices is given by  circulant matrices.
\begin{definition}
 An $r \times r$ matrix $A=(A_{i,j})$ over a field $\F$ is called \emph{circulant matrix} if there exists a vector $a=(a_0,\ldots, a_{r-1})$ such that 
$$A_{i,j}=a_{j-i (\!\!\!\!\mod r)}.$$
\end{definition}

\begin{theorem}\label{thm:Hank}
For every $0<k<n \leq m$ and every $s$ coprime to $m$, there exists a  code $\C_X \in \Gab_q(k,n,m,s)$ such that $X$ is a Hankel matrix.
\end{theorem}

\begin{proof}
Let $\gamma \in \Fm$ be a primitive element. By Lemma \ref{lem:primtrace} there exist $\ell \in \mathbb N$ such that 
\begin{equation}\label{eq:Tr=0}\Tr(\gamma^\ell)=\Tr(\gamma^{\ell+1})=\ldots=\Tr(\gamma^{\ell+m-2})=0.\end{equation}
Let $\bs\alpha \in \Fm^k$, $\bs\beta \in \Fm^{n-k}$  be defined as
\begin{align*}
\bs\alpha&=(\gamma^\ell, \gamma^{\ell+1},\ldots,\gamma^{\ell+k-1}),\\
\bs\beta&=(1,\gamma,\ldots,\gamma^{n-k-1}),
\end{align*}
and consider the matrix $\bs\alpha^\top\bs\beta$. We now check that $\bs\alpha ,\bs\beta$ satisfy properties (a), (b), (c) of Theorem \ref{thm:stform}. Indeed, $\gamma$ is a primitive element, and therefore $1,\gamma,\ldots,\gamma^{m-1}$ are linearly independent, as well as $\gamma^\ell,\ldots,\gamma^{\ell+m-1}$. In particular, properties (a) and (b) are satisfied. Moreover,  for every $i=0,\ldots,k-1$, $j=0,\ldots,n-k-1$, 
$$T_{\gamma^{\ell+i}}(\gamma^j)=\Tr(\gamma^{\ell+i+j})=0,$$
where the last inequality holds by (\ref{eq:Tr=0}). Therefore also property (c) is verified.

Now we have that every matrix $X\in \Phi_s^{-1}(\{\bs\alpha^\top\bs\beta\})$ is a $(q,s)$-Cauchy matrix and hence it is of the form
$$ X=\begin{pmatrix} \pi_s(\gamma^{\ell}) & \pi_s(\gamma^{\ell+1}) & \cdots & \pi_s(\gamma^{\ell+n-k-1})\\
\pi_s(\gamma^{\ell+1}) & \pi_s(\gamma^{\ell+2}) & \cdots & \pi_s(\gamma^{\ell+n-k})\\
\vdots & \vdots &  & \vdots \\
\pi_s(\gamma^{\ell+k-1}) & \pi_s(\gamma^{\ell+k}) & \cdots & \pi_s(\gamma^{\ell+n-2})\\
\end{pmatrix}+B,$$
for an arbitrary $B\in\F_q^{k\times(n-k)}$. Choosing $B$ as a Hankel matrix completes the proof.
\end{proof}

\begin{theorem}\label{thm:Toeplitz}
For every $0<k<n \leq m$ and every $s$ coprime to $m$, there exists a generalized Gabidulin code $\C_X$ of parameter $s$ such that $X$ is a Toeplitz matrix.
\end{theorem}

\begin{proof}
 Following the same proof of Theorem \ref{thm:Hank} with 
\begin{align*}
\bs\alpha&=(\gamma^{\ell+n-k-1}, \gamma^{\ell+n-k},\ldots,\gamma^{\ell+n-2}),\\
\bs\beta&=(1,\gamma^{-1},\gamma^{-2},\ldots,\gamma^{-n+k+1}),
\end{align*}
the matrix obtained is of the form
$$ X=\begin{pmatrix} \pi_s(\gamma^{\ell+n-k-1}) & \pi_s(\gamma^{\ell+n-k-2}) & \cdots & \pi_s(\gamma^{\ell})\\
\pi_s(\gamma^{\ell+n-k}) & \pi_s(\gamma^{\ell+n-k-1}) & \cdots & \pi_s(\gamma^{\ell+1})\\
\vdots & \vdots &  & \vdots \\
\pi_s(\gamma^{\ell+n-2}) & \pi_s(\gamma^{\ell+n-3}) & \cdots & \pi_s(\gamma^{\ell+k-1})\\
\end{pmatrix}+B,$$
for an arbitrary $B\in\F_q^{k\times(n-k)}$. As above, choosing $B$ in Toeplitz form concludes the proof.
\end{proof}

These two theorems allow to define two subfamilies of generalized Gabidulin codes, the \emph{Hankel Gabidulin codes} and the \emph{Toeplitz Gabidulin codes}. In the following Lemma we can see that this structure on the generator matrix in standard form is hard to improve if we still require the code to be MRD.

\begin{lemma}\label{lem:circulant}
Suppose that $n$ is even and $k=\frac{n}{2}$. Let $X\in \Fm^{k \times k}$ be a circulant matrix, and let $d$ be the minimum rank distance of the code $\C_X$. Then $d \leq 2$. 

In particular, for $n\geq 4$, there does not exist any $\frac{n}{2}$-dimensional MRD code $\C_X$ with $X$ circulant matrix.
\end{lemma}

\begin{proof}
Since the matrix $X$ is circulant, then the sum of the elements on each of its columns is constant. Let $\gamma$ be such a sum. Then, the non-zero codeword of the code $\C_X$
$$(1,\ldots, 1)  \left(\begin{array}{c|c}
      I_k \; & X
    \end{array}
  \right)=(1,\ldots, 1, \gamma, \ldots, \gamma)$$
has rank weight at most $2$. In particular, if $n\geq 4$ we have $$n-k+1=\frac{n}{2}+1>2 \geq d$$ and therefore, the code $\C_X$ can not be MRD.
\end{proof}

This result possibly suggests that, at least in the case $k=\frac{n}{2}$, it is very difficult to require more structure on the non-systematic part of the generator matrix in standard form of an MRD code. However, it would be very interesting to find new families of Gabidulin, or more generally MRD codes with structured generator matrices.

We conclude this section with a small  example.

\begin{example}
Consider the case $q=2$, $k=3$, $n=m=6$ and $s=1$. We construct a Hankel Gabidulin code of dimension $3$ and length $6$ over the finite field $\F_{2^6}=\F_2(a)$, where $a$ is a primitive element that satisfies $a^6+a^4+a^3+a+1=0$. One can find, by Lemma \ref{lem:primtrace}, five consecutive powers of $a$ that belong to $\ker( \mathrm{Tr}_{\F_{2^6}/\F_2})$, that are $a^i$ for $i=14,15,\ldots, 18$. Then, we set the vectors
$$\bs\alpha=(a^{14},a^{15},a^{16}), \quad \bs\beta=(1,a,a^2).$$
Moreover, we choose the matrix $B$ to be the zero matrix, and the map

$$ \begin{array}{ >{\displaystyle}r >{{}}c<{{}}  >{\displaystyle}l } 
\pi_1 : \F_{2^6} & \longrightarrow & \F_{2^6} \\
z & \longmapsto & 
 \sum_{i=0}^{4}\left(\gamma^{2^{i+1}}\sum_{j=0}^iz^{2^j}\right),\end{array}$$
 with $\gamma=a^3$. We then get the following $q$-Cauchy matrix
$$X:=C_{(q,1)}(\bs\alpha,\bs\beta, 0)=\begin{pmatrix}\pi_1(a^{14}) &  \pi_1(a^{15})  & \pi_1(a^{16})  \\
\pi_1(a^{15})  & \pi_1(a^{16})  & \pi_1(a^{17})  \\
\pi_1(a^{16})  & \pi_1(a^{17})  & \pi_1(a^{18}) \end{pmatrix}=\begin{pmatrix} a^{57} & a^{7} & a^{13} \\
a^{7} & a^{13} & a^{37} \\
a^{13} & a^{37} & a^{36} \end{pmatrix}.$$
By Theorem \ref{thm:Hank} the code $\C_X$ is a Gabidulin code of parameter $s=1$. Moreover we can recover the vector of evaluation points $g=(g_1,\ldots, g_6)$ of the code. We can suppose $g_1=1$, and recover $g_2=a^{45}$ and $g_3=a^{!5}$ using the linear system (\ref{eq:linsyst}). Finally, using equations (\ref{eq:first}) and (\ref{eq:firstgen}) we get $g_4=a^{46}$, $g_5=a^{14}$ and $g_6=a^{28}$. Therefore, our Hankel Gabidulin code is
$$\C_X=\mathcal G_{3,1}(1,a^{45},a^{15},a^{46},a^{14},a^{28}).$$
\end{example}

\section{Conclusions and open problems}

In this work we find a parametrization of the generator matrix in standard form of generalized Gabidulin codes (see Theorem \ref{thm:stform}). Such a parametrization coincides with the $q$-analogue of generalized Cauchy matrices and, therefore, leads to a natural definition of \emph{$(q,s)$-Cauchy matrices}, a notion that was missing in the literature. In Theorem \ref{thm:stformbis} we underline that these matrices are in one-to-one correspondence with generalized Gabidulin codes.
Moreover, in Theorems \ref{thm:GabCritNew} and \ref{thm:GabCritNewII} we give a new criterion for determining whether a given rank metric code of dimension $k$ and length $n$ is a generalized Gabidulin code. This new result only requires $\mathcal O(m \cdot F(k,n))$ field operations, where $F(k,n)$ denotes the cost of computing the reduced row echelon form of a $k\times n$ matrix, and it improves the existing criterion that relies on an a priori check of the MRD property. 
Finally we use our results in order to build two new subfamilies of generalized Gabidulin codes, namely the \emph{Hankel and Toeplitz Gabidulin codes} (see Theorems \ref{thm:Hank} and \ref{thm:Toeplitz}). These families have the advantage that the non-systematic part of the generator matrix is a structured matrix. This implies that matrix/vector multiplications, and therefore the encoding procedure, can be performed faster than usual. 

 From a theoretical point of view we believe that the same parametrization as the one of Theorems \ref{thm:stform} and \ref{thm:stformbis} applies to Gabidulin codes over any finite cyclic field extension $\mathbb{E}/\mathbb{K}$, which were introduced by Augot, Loidreau and Robert in \cite{augot2013rank, augot2014generalization, augot2018generalized} (see also \cite[Section VI]{roth1996tensor}). As a consequence one would also get the characterization of generalized Gabidulin codes of Theorems \ref{thm:GabCritNew} and \ref{thm:GabCritNewII} in this setting. This general approach based on general cyclic extension field is used to describe some results of this paper in \cite[Chapter 4]{neri2019PhD}. Here, $\sigma$-Gabidulin codes are defined with respect to a generator $\sigma$ of the Galois group $G:=\Gal(\mathbb{E}/\mathbb{K})$, and the space $\mathcal G_{k,s}$ is replaced by the space $\langle \mathrm{id}, \sigma, \ldots, \sigma^{k-1}\rangle_{\mathbb{E}}\subseteq \mathbb{E}[G]$. Formally we have the following open problem.
 
 \begin{problem}
  Let $\mathbb{E}/\mathbb{K}$ be an extension of fields of finite degree, and let $\Gal(\mathbb{E}/\mathbb{K})$ be a cyclic group. Do Theorems \ref{thm:stform} and \ref{thm:stformbis} hold in this more general setting?
 \end{problem}
 
 Unfortunately, the proof technique used here relies on a counting argument, and therefore it does not apply to infinite fields. In order to prove the same theorems in this more general setting one needs to find a different argument, finding an explicit bijection between generalized Gabidulin codes and $(q,s)$-Cauchy matrices which also works for infinite fields.
 
 Moreover, for future research we plan to investigate  the use of this parametrization for  applications in erasure and syndrome decoding for generalized Gabidulin codes.

\section*{Acknowledgement} 
The author would like to thank Eimear Byrne for useful comments on the structure of the work and  Martino Borello for suggesting to point out Lemma \ref{lem:circulant}.

\bibliographystyle{abbrv}
\bibliography{./network_coding_stuff}

\end{document}